\documentclass[11pt]{article}

\def\showauthornotes{1}
\def\showdraftbox{0}

\newcommand{\deeksha}{\Authornote{Deeksha}}

\usepackage{amstext,amsfonts, amsmath, amssymb,amsthm, bbm,xspace,nicefrac,graphicx, url,array}
\usepackage{color,enumerate,latexsym,bm,verbatim,textcomp}

\definecolor{ForestGreen}{rgb}{0.1333,0.5451,0.1333}
\definecolor{DarkRed}{rgb}{0.65,0,0}
\definecolor{Red}{rgb}{1,0,0}
 \usepackage[linktocpage=true,colorlinks,
 linkcolor=DarkRed,citecolor=ForestGreen,
 bookmarks,bookmarksopen,bookmarksnumbered]
 {hyperref}
\usepackage{cleveref}

\usepackage[small]{caption}
\usepackage{comment} 
\usepackage{algorithm} 
\usepackage[noend]{algpseudocode}
\usepackage{thmtools}
\usepackage{thm-restate}
\usepackage{fancybox}
\usepackage[shortlabels]{enumitem}
\usepackage{fullpage}
\usepackage{booktabs}
\usepackage{commath}
\usepackage{mdframed}
\usepackage{pdfsync}

\newcommand{\defeq}{\stackrel{\textup{def}}{=}}

\newtheorem{theorem}{Theorem}[section]
\newtheorem{problem}[theorem]{Problem}
\newtheorem{lemma}[theorem]{Lemma}
\newtheorem{corollary}[theorem]{Corollary}

\newtheorem{property}[theorem]{Property}
\newtheorem{proposition}[theorem]{Proposition}
\newtheorem{hypothesis}[theorem]{Hypothesis}

\newtheorem{claim}[theorem]{Claim}

\newtheorem{definition}[theorem]{Definition}

\DeclareMathOperator*{\av}{\mathbbm{E}}

\renewcommand{\norm}[1]{\ensuremath{\left\lVert #1 \right\rVert}}

\newcommand{\marginlabel}[1]%
{\mbox{}\marginpar{\it{\raggedleft\hspace{0pt}#1}}}

\newcommand\poly{\mathrm{poly}}  

\newcommand\calM{\mathcal{M}}

\DeclareMathOperator{\Tr}{Tr}

\definecolor{Mygray}{gray}{0.8}

 \ifcsname ifcommentflag\endcsname\else
  \expandafter\let\csname ifcommentflag\expandafter\endcsname
                  \csname iffalse\endcsname
\fi

\ifnum\showauthornotes=1
\newcommand{\Authornote}[2]{{\sf\small\color{red}{[#1: #2]}}}
\newcommand{\Authoredit}[2]{{\sf\small\color{red}{[#1]}\color{blue}{#2}}}
\newcommand{\Authorcomment}[2]{{\sf \small\color{gray}{[#1: #2]}}}
\newcommand{\Authorfnote}[2]{\footnote{\color{red}{#1: #2}}}
\newcommand{\Authorfixme}[1]{\Authornote{#1}{\textbf{??}}}
\newcommand{\Authormarginmark}[1]{\marginpar{\textcolor{red}{\fbox{
#1:!}}}}
\else
\newcommand{\Authornote}[2]{}
\newcommand{\Authoredit}[2]{}
\newcommand{\Authorcomment}[2]{}
\newcommand{\Authorfnote}[2]{}
\newcommand{\Authorfixme}[1]{}
\newcommand{\Authormarginmark}[1]{}
\fi

\newlength{\pgmtab}  
\let\originalleft\left
\let\originalright\right
\renewcommand{\left}{\mathopen{}\mathclose\bgroup\originalleft}
  \renewcommand{\right}{\aftergroup\egroup\originalright}

\def\defeq{\stackrel{\mathrm{def}}{=}}

\def\dim#1{\mathrm{dim} (#1)}
\def\dimm{\mathrm{dim}}

\newcommand\ff{\boldsymbol{\mathit{f}}}

\newcommand\uu{\boldsymbol{\mathit{u}}}
\newcommand\vv{\boldsymbol{\mathit{v}}}
\newcommand\ww{\boldsymbol{\mathit{w}}}
\newcommand\yy{\boldsymbol{\mathit{y}}}
\newcommand\zz{\boldsymbol{\mathit{z}}}
\newcommand\xx{\boldsymbol{\mathit{x}}}

\renewcommand\AA{\boldsymbol{\mathit{A}}}
\newcommand\BB{\boldsymbol{\mathit{B}}}
\newcommand\CC{\boldsymbol{\mathit{C}}}

\newcommand\II{\boldsymbol{\mathit{I}}}

\newcommand\LL{\boldsymbol{\mathit{L}}}
\newcommand\PP{\boldsymbol{\mathit{P}}}
\newcommand\QQ{\boldsymbol{\mathit{Q}}}

\renewcommand\SS{\boldsymbol{\mathit{S}}}
\newcommand\UU{\boldsymbol{\mathit{U}}}
\newcommand\WW{\boldsymbol{\mathit{W}}}
\newcommand\VV{\boldsymbol{\mathit{V}}}
\newcommand\XX{\boldsymbol{\mathit{X}}}
\newcommand\YY{\boldsymbol{\mathit{Y}}}

\newcommand\Ttil{{\tilde{\mathit{T}}}}
\newcommand\AAtil{\boldsymbol{\widetilde{\mathit{A}}}}
\global\long\def\Phitil{\tilde{\Phi}}

\newcommand\Otil{\widetilde{O}}

 {
	\begin{enumerate}}{\end{enumerate}}

\ifnum\showdraftbox=1

\else

\fi

\global\long\def\nnz{\mathrm{nnz}}

\global\long\def\eps{\epsilon}

\newcommand\Span{\mathrm{span}}
\newcommand\Var{\mathrm{Var}}

\usepackage[
backend=biber,
backref=true,
backrefstyle=none,
date=year,
doi=false,
giveninits=true,
hyperref=true,
maxbibnames=10,
sortcites=false,
style=alphabetic,
url=false, 
]{biblatex}
\begin{document}

	\title{Decremental $(1+\epsilon)$-Approximate Maximum Eigenvector:\\Dynamic Power Method}
	\author{Deeksha Adil\thanks{Supported by Dr. Max R\"ossler, the Walter Haefner Foundation and the ETH Z\"urich Foundation}\\Institute for Theoretical Studies\\ETH Zürich\\deeksha.adil@eth-its.ethz.ch \and Thatchaphol Saranurak\thanks{Supported by NSF grant CCF-2238138.}\\Department of Computer Science\\University of Michigan\\thsa@umich.edu}
	\maketitle

	\maketitle

	\pagenumbering{gobble}
	\begin{abstract}
		We present a dynamic algorithm for maintaining $(1+\epsilon)$-approximate maximum eigenvector and eigenvalue of a positive semi-definite matrix $A$ undergoing \emph{decreasing} updates, i.e., updates which may only decrease eigenvalues. Given a vector $v$ updating $A\gets A-vv^{\top}$, our algorithm takes $\tilde{O}(\mathrm{nnz}(v))$ amortized update time, i.e., polylogarithmic per non-zeros in the update vector.
		
		Our technique is based on a novel analysis of the influential power method in the dynamic setting. The two previous sets of techniques have the following drawbacks (1) algebraic techniques can maintain exact solutions but their update time is at least polynomial per non-zeros, and (2) sketching techniques admit polylogarithmic update time but suffer from a crude additive approximation.
		
		Our algorithm exploits an oblivious adversary. 
        Interestingly, we show that any algorithm with polylogarithmic update time per non-zeros that works against an adaptive adversary and satisfies an additional natural property would imply a breakthrough for checking psd-ness of matrices in $\tilde{O}(n^{2})$ time, instead of $O(n^{\omega})$ time.
	\end{abstract}
	
	\newpage
	\pagenumbering{gobble}
	\tableofcontents
	\newpage
	
	\pagenumbering{arabic}

\section{Introduction}

Computing eigenvalues and eigenvectors of a matrix is predominant in several applications, including principle component analysis, clustering in high-dimensional data, semidefinite programming, spectral graph partitioning, and algorithms such as Google's PageRank algorithm. In the era of massive and dynamic datasets, developing algorithms capable of efficiently updating spectral information with changing inputs has become indispensable.

The study of the change in the spectrum of a matrix undergoing updates spans over five decades, with \textcite{golub1973some} providing the first algebraic characterization of the change for positive semi-definite matrices via the {\it secular equation}. This result offered an explicit formula for exactly computing all new eigenvalues when a matrix undergoes a single rank-one update, assuming knowledge of the entire eigen-decomposition of the initial matrix. Subsequently, \textcite{bunch1978rank} showed how to additionally compute eigenvectors explicitly, and handle the case of repeated eigenvalues of the initial matrix. A decade later, \textcite{arbenz1988restricted} extended the works of \cite{golub1973some,bunch1978rank} to compute all new eigenvalues and eigenvectors of a positive semi-definite matrix undergoing an update of the form of a small-rank matrix, or a single batch update. Further works extend these techniques to computing singular values \cite{stange2008efficient}, as well as computing new eigenvalues and eigenvectors when only the partial eigen-decomposition of the initial matrix is known \cite{mitz2019symmetric}. However, all these works require a full eigen-decomposition of the initial matrix and only handle a single update to the matrix. 

Another independent line of work aims at finding a small and sparse matrix that has eigenvalues close to the original matrix. 
A recent work in this direction by \textcite{bhattacharjee2023universal} provides a universal sparsifier $\SS$ for positive semi-definite matrices $\AA$, such that $\SS$ has at most $n/\epsilon^4$ non-zero entries and $\|\AA-\AA\circ\SS\|\leq \epsilon n$. \textcite{swartworth2023optimal} give an alternate technique by using {\it gaussian sketches} with $O(1/\epsilon^2)$ rows to approximate all eigenvalues of $\AA$ to an additive error of $\epsilon\|\AA\|_F$. The algorithms that find sparse matrices approximating the eigenvalues in these works may be extended to handle updates in the initial matrix quickly. However, the approximation of the eigenvalues achieved is quite crude, i.e., at least $\epsilon n$ additive factor. This approximation is also shown to be tight for such techniques.

Now consider the simpler task of only maintaining the maximum eigenvalue and eigenvector of a matrix $\AA$ as it undergoes updates. As we have seen above, known methods based on algebraic techniques require full spectral information before computing the new eigenvalues, and maintaining the entire eigen-decomposition can be slow. Sparsification-based algorithms are fast but only achieve large additive approximations of at least $\epsilon n$, which is not quite satisfactory. Works on streaming PCA, which have a similar goal of maintaining large eigenvectors focus on optimizing the space complexity instead of runtimes \cite{allen2017first}.
Previous works based on dynamically maintaining the matrix inverse can maintain $(1+\eps)$-multiplicative approximations to the maximum eigenvalues of an $n \times n$ symmetric matrix undergoing single-entry updates with an update time of $O(n^{1.447}/\poly(\eps))$ \cite{sankowski2004dynamic}. This was further improved to $O(n^{1.407}/\poly(\eps))$\cite{van2019dynamic}. The update time is even slower for row, column, or rank-one updates. In any case, it takes polynomial time per non-zeros of the updates. This leads to the following natural question:

\begin{center}
    \emph{Is there a dynamic algorithm that can maintain a multiplicative approximation of the maximum eigenvalue of a matrix using polylogarithmic update time per non-zeros in the update?}\par
\end{center}

In this paper, we study the problem of maintaining a $(1+\epsilon)$-multiplicative approximation to the maximum eigenvalue and eigenvector of a positive semi-definite matrix as it undergoes a sequence of rank-one updates that may only decrease the eigenvalues. We note that this is equivalent to finding the minimum eigenvalue and eigenvector of a positive semi-definite matrix undergoing rank-one updates that may only increase the eigenvalues.
We now formally state our problem.
\begin{problem}
	[Decremental Approximate Maximum Eigenvalue and Eigenvector]\label{prob:dyn}We are given a size parameter $n$, an accuracy parameter $\epsilon\in(0,1)$, a psd matrix $\AA_{0}\succeq0$ of size $n\times n$, and an online sequence of vectors $\vv_{1},\vv_{2},\dots,\vv_{T}$ that update $\AA_{t}\gets \AA_{t-1}-\vv_{t}\vv_{t}^{\top}$ with a promise that $\AA_{t}\succeq0$ for all $t$. 
	The goal is to maintain an $\epsilon$-approximate eigenvalue $\lambda_{t}\in\mathbb{R}$ and $\epsilon$-approximate eigenvector $\ww_{t}\in\mathbb{R}^{n}$ of $\AA_{t}$ for all time $t$. That is $\lambda_t$ and unite vector $\ww_t$ satisfying,
	\begin{equation}\label{eq:epsEigvalue}
	\max_{\uu:\textrm{unit}}\uu^{\top}\AA_{t}\uu \ge \lambda_{t} \ge (1-\epsilon)\max_{\uu:\textrm{unit}}\uu^{\top}\AA_{t}\uu,
	\end{equation}
	and
	\begin{equation}\label{eq:epsEigvec}
	\ww_{t}^{\top}\AA_{t}\ww_{t}\ge(1-\epsilon)\max_{\uu:\textrm{unit}}\uu^{\top}\AA_{t}\uu.
	\end{equation}
\end{problem}

\subsection{Our Results}

We give an algorithm for the Decremental Approximate Maximum Eigenvalue and Eigenvector Problem that has an amortized update time of $\approx \Otil(nnz(\vv_t))\leq \Otil(n)$\footnote{$\Otil$ hides polynomials in $\log n$.}. In other words, the total time required by our algorithm over $T$ updates is at most $\Otil(nnz(\AA_0) + \sum_{i=1}^T nnz(\vv_i))$. Observe that our algorithm only requires time that is $\Otil(1) \times$ (time required to read the input) and can handle any number of decremental updates. Our algorithm works against an {\it oblivious} adversary, i.e., the update sequence is fixed from the beginning. This is the first algorithm that can handle a sequence of updates while providing a multiplicative approximation to the eigenvalues and eigenvectors in a total runtime of $\leq \Otil(n^2 + n\cdot T)$ and an amortized update time faster than previous algorithms by a factor of $n^{\Omega(1)}$. Formally, we prove the following:
\begin{theorem}
	\label{thm:upper}There is an algorithm for \Cref{prob:dyn} under a sequence of $T$ decreasing updates, that given parameters $n$, $\AA_0$, and $\epsilon>1/n$ as input, with probability at least $1-1/n$ works against an oblivious adversary in total time,
 \[
 O\left(\frac{\log^{3}n \log^6 \frac{n}{\epsilon}\log \frac{\lambda_{\max}(\AA_0)}{\lambda_{\max}(\AA_T)}}{\epsilon^{4}}\left(nnz(\AA_0) + \sum_{i=1}^T nnz(\vv_i)\right) \right).
 \]
\end{theorem}

 Our algorithm is a novel adaptation of the classical {\it power method} (see, e.g., \cite{trefethen2022numerical}) 
 to the dynamic setting, along with a new analysis that may be of independent interest. 

 Our work can also be viewed as the first step towards generalizing the dynamic algorithms for solving positive linear programs of \cite{bhattacharya2023dynamic} to solving dynamic positive semi-definite programs. We discuss this connection in detail in \Cref{sec:dyn psd}.

\subsection{Towards Separation between Oblivious and Adaptive adversaries}
 We also explore the possibility of removing the assumption of an oblivious adversary and working against {\it adaptive adversaries}. Recall that update sequences are given by adaptive adversaries when the update sequence can depend on the solution of the algorithm. 
 
 We show that if there is an algorithm for \Cref{prob:dyn} with a total running time of at most $\Otil(n^2)$ such that the output $\ww_t$'s satisfy an additional natural property (which is satisfied by the output of the power method, but \emph{not} by our algorithm), then it contradicts the hardness of a well-known barrier in numerical linear algebra, therefore ruling out such algorithms.
 
 We first state the barrier formally, and then state our result for adaptive adversaries.
Recall that, given a matrix $\AA$, the condition number of $\AA$ is $\frac{\max_{x:\textrm{unit}}\|\AA\xx\|}{\min_{x:\textrm{unit}}\|\AA\xx\|}$.

\begin{problem}[Checking psdness with certificate]\label{prob:factor}Given $\delta\in(0,1)$, parameter $\kappa$, and a symmetric matrix $\AA$ of size $n\times n$ with condition number at most $\kappa$, either 
	\begin{itemize}
		\item Compute a matrix $\XX$ where $\|\AA-\XX\XX^{T}\|\le\delta\min_{\|\xx\|=1}\|\AA\xx\|$, certifying that $\AA$ is a psd matrix, or
		\item Report that $\AA$ is not a psd matrix.
	\end{itemize}
\end{problem}

Recall that $A$ is a psd matrix iff there exists $\XX$ such that $\AA=\XX\XX^{\top}$. The matrix of $\XX$ is called a \emph{vector realization} of $\AA$, and note that $\XX$ is not unique. The problem above asks to compute a (highly accurate) vector realization of $\AA$. Both eigendecomposition and Cholesky decomposition of $\AA$ give a solution for \Cref{prob:factor}.\footnote{Given an eigendecomposition $\AA=\QQ\Lambda \QQ^{\top}$ where $\Lambda$ is a non-negative diagonal matrix and $\QQ$ is an orthogonal matrix, we set $\XX=\QQ\Lambda^{1/2}$. Given a Cholesky decomposition $\AA=\LL\LL^{\top}$ where $\LL$ is a lower triangular matrix, we set $\XX=\LL$.} Banks et al~\cite{banks2022pseudospectral} showed how to compute with high accuracy an eigendecomposition of a psd matrix in $O(n^{\omega+\eta})$ time for any constant $\eta>0$, where $\omega > 2.37$ is the matrix multiplication exponent. Hence, a vector realization of $\AA$ can be found in the same time. Observe that, when $\delta < \kappa$, \Cref{prob:factor} is \emph{at least as hard as certifying that a given matrix $\AA$ is a psd matrix}.
It is a notorious open problem whether certifying the psd-ness of a matrix can be done faster than $o(n^\omega)$ even when the condition number is polynomial in $n$.
Therefore, we view \Cref{prob:factor} as a significant barrier and formalize it as follows.%

\begin{hypothesis}
	[PSDness Checking Barrier]\label{conj:decomp is hard}There is a constant $\eta>0$ such that, every randomized algorithm for solving \Cref{prob:factor} for instances with $\frac{\kappa}{\delta}\leq \poly(n)$, correctly with probability at least $2/3$ requires $n^{2+\eta}$ time. 
\end{hypothesis}

Our negative result states that, assuming  Hypothesis~\ref{conj:decomp is hard}, there is no  algorithm for Problem~\ref{prob:dyn} against adaptive adversaries with sub-polynomial update time
that maintains $\ww_t$ satisfying an additional property stated below.

\begin{property}\label{def:super}\label{prop:assume} For every $t$ let $\uu_i(\AA_t)$ denote the eigenvectors of $\AA_t$ and $\lambda_i(\AA_t)$ denote the eigenvalues. For all $i$ such that $\lambda_i(\AA_t)\leq \lambda_1(\AA_t)/2$,
\[
\left(\ww_t^{\top}\uu_i\left(\AA_{t}\right)\right)^2 \leq \frac{1}{n^2}\cdot \frac{\lambda_i\left(\AA_{t}\right)}{\lambda_1(\AA_t)}.
\]
\end{property}

\begin{theorem}
	\label{thm:lower}Assuming \Cref{conj:decomp is hard}, there is no algorithm for \Cref{prob:dyn} that maintains $\ww_t$'s additionally satisfying \Cref{def:super}, and given parameters $n$ and $\epsilon=\min\left\{1-\frac{1}{n^{o(1)}},\frac{1-\delta}{1+\delta}\right\}$ as input, works against an adaptive adversary in time $n^{o(1)}\cdot \left(nnz(\AA_0) + \sum_{t=1}^T nnz(\vv_i)\right)$. 
\end{theorem}

Let us motivate \Cref{def:super} from several aspects below. First, in the static setting, this property can be easily satisfied since the static power method strongly satisfies it (see \Cref{thm:StaticPowerMain}). 
Second, the statement of the property itself is a natural property that we might expect from an algorithm. Consider a ``bad'' eigenvector $\uu_i$ whose corresponding eigenvalue is very small, i.e., less than half of the maximum one.
It states that the projection of the output $\ww_t$ along $\uu_i$ should be very small, i.e., a polynomial factor smaller than the projection along the maximum eigenvector. This is intuitively useful because we do not want the output $\ww_t$ to direct to the ``bad'' direction. It should mainly direct along the approximately maximum eigenvector.
Third, we formalize this intuition and crucially exploit \Cref{def:super} to prove a reduction in \Cref{thm:lower}. 
Specifically, \Cref{def:super} allows us to decrease eigenvalues of a psd matrix while maintaining the psd-ness, which is crucial for us. See \Cref{sec:Adap} for more details. 
Lastly, our current algorithm from \Cref{thm:upper}, a dynamic version of the power method, actually maintains $\ww_t$'s that satisfy this property for certain snapshots, but not at every step. Unfortunately, we do not see how to strengthen our algorithm to satisfy \Cref{prop:assume} nor how to remove it from the requirement of \Cref{thm:lower}. We leave both possibilities as open problems.

Understanding the power and limitations of dynamic algorithms against an oblivious adversary vs.~an adaptive adversary has become one of the main research programs in the area of dynamic algorithms. However, finding a natural dynamic problem that separates oblivious and adaptive adversaries is still a wide-open problem.\footnote{
Beimel et al.~\cite{beimel2022dynamic} gave the first separations for artificial problems assuming a strong cryptographic assumption. Bateni et al.~\cite{bateni2023optimal} shows a separation between the adversaries for the {\it $k$-center clustering problem}, however, their results are based on the existence of a black box which the adaptive adversary can control.}
%
%
In this paper, we suggest the problem of maintaining approximate maximum eigenvectors as a natural candidate for the separation and 
give some preliminary evidence for this.

\paragraph{Organization.} In the following sections, we begin with some preliminaries required to prove our results in Section~\ref{sec:prelims}, followed by our algorithm and its analysis in Section~\ref{sec:Obl}, and our conditional lower bound in Section~\ref{sec:Adap}. In Section~\ref{sec:dyn psd} we show connections with positive semi-definite programs, and finally, in Section~\ref{sec:open}, we present some open problems.


\section{Preliminaries}
\label{sec:prelims}
Let $\AA_0$ denote the initial matrix. Let $\lambda_0 = \lambda_{\max}(\AA_{0}) = \|\AA_{0}\|$ denote the maximum eigenvalue of the initial matrix $\AA_0$.
The following are the key definitions we use in our analysis.

\begin{definition}[$\epsilon$-max span and dimension] We define $\Span(\epsilon,\AA)$ to denote the space spanned by all eigenvectors of $\AA$ corresponding to eigenvalues $\lambda$ satisfying $\lambda \geq (1-\epsilon)\lambda_0$. Let $\dim{\epsilon,\AA}$ to denote the dimension of the space $\Span(\epsilon,\AA)$.
\end{definition}

We emphasize that $\lambda_0$ depends only on $\AA_0$. So, it is a static value that does not change through time. 
We will use the following linear algebraic notations.

\begin{definition}\label{def:subspace}Let $S_1$ and $S_2$ be two subspaces of a vector space $S$. The sum $S_1+S_2$ is the space,
\[
S_1+S_2 = \{s=s_1+s_2: s_1\in S_1, s_2 \in S_2\}.
\]
The complement $\overline{S}$ of $S$ is the vector space such that $S+\overline{S} = \mathbb{R}^n$, and $S\cap \overline{S} = \{0\}$. The difference, $S_1-S_2$ is defined as,
\[
S_1-S_2 = S_1\cap \overline{S_2}.
\]
\end{definition}

Next, we list standard facts about high-dimensional probability needed in our analysis.

\begin{lemma}[Chernoff Bound]\label{lem:Bernstein} Let $x_1,\cdots x_m$ be independent random variables such that $a\leq x_i\leq b$ for all $i$. Let $x = \sum_i x_i$ and let $\mu = \av[x]$. Then for all $\delta>0$,
\[
\Pr[x\geq(1+ \delta) \mu] \leq \exp \left(- \frac{2\delta^2\mu^2 }{m(b-a)^2} \right)
\]
\[
\Pr[x \leq (1-\delta) \mu] \leq \exp \left(- \frac{\delta^2\mu^2 }{m(b-a)^2} \right)
\]
\end{lemma}

\begin{lemma}[Norm of Gaussian Vector]\label{lem:NormG}
A random vector $\vv \in \mathbb{R}^n$ with every coordinate chosen from a normal distribution, $N(0,1)$ satisfies,
\[
\Pr[|\|\vv\|^2- n| \leq 2(1+\delta)\delta\cdot n ] \geq 1- e^{-\delta^2  n}.
\]
\end{lemma}
\begin{proof}
The vector $\vv$ has entries that are from $N(0,1)$. Now, every $\vv_i^2$ follows a $\chi^2$ distribution. From Lemma 1 of \cite{laurent2000adaptive} we have the following tail bound for a sum of $\chi^2$ random variables,
\[
\Pr[|\sum_i \vv^2_i - n| > 2\sqrt{nx} + 2x]\leq e^{-x}.
\]
Choosing $x = n\delta^2$ gives,
\[
\Pr[|\|\vv\|^2 - n| \leq 2\delta(1+\delta) n]\geq 1-e^{-\delta^2 n},
\] 
as required.
\end{proof}

\begin{lemma}[Distribution of $\chi^2$ Variable]\label{lem:chi}
Let $x\sim N(0,1)$ be a gaussian random variable. Then,
\[
\Pr\left[x^2\geq \frac{1}{n^4}\right] \geq 1- \frac{1}{n^{2}}.
\]
\end{lemma}
\begin{proof}
The probability distribution function for $y = x^2$ is given by,
\[
\ff(y) = \frac{1}{\sqrt{2}\Gamma(\frac{1}{2})}y^{-\frac{1}{2}}e^{-\frac{y}{2}}.
\]
It is known that $\Gamma(\frac{1}{2}) = \sqrt{\pi}$. Now, 
\begin{align*}
\Pr\left[x^2\leq \frac{1}{n^4}\right] =  \int_{0}^{1/n^4}\frac{1}{\sqrt{2}\Gamma(\frac{1}{2})}y^{-\frac{1}{2}}e^{-\frac{y}{2}}dy
\leq  \frac{e^0}{\sqrt{2\pi}} \int_{0}^{1/n^4} y^{-\frac{1}{2}}dy
=  \sqrt{\frac{2}{\pi}} \cdot \frac{1}{n^{2}}
\leq  \frac{1}{n^{2}}.
\end{align*}
\end{proof}


\section{Algorithms against an Oblivious Adversary}\label{sec:Obl}

To prove \Cref{thm:upper}, we first reduce Problem~\ref{prob:dyn} to solving a normalized threshold version of the problem where we assume that initially, the maximum eigenvalue is not much bigger than one. Then we want to maintain a certificate that 
the maximum eigenvalue is not much less than one until no such certificate exists. This is formalized below.
\begin{problem}[DecMaxEV($\epsilon,\AA_0,\vv_1,\cdots,\vv_T$)]\label{def:DecMaxEval} Let $\AA_0$ be an $n\times n$ symmetric PSD matrix such that $ \lambda_{\max}(\AA_0) \leq 1 + \frac{\epsilon}{\log n}$. The {\sc DecMaxEV}($\epsilon,\AA_0,\vv_1,\cdots,\vv_T$) problem asks to find for every $t$, a vector $\ww_t$ such that 
\[
\|\ww_t\| = 1 \quad \text{and} \quad \ww_t^{\top}\AA_t \ww_t \geq 1-40\epsilon,
\]
or return {\sc False} indicating that $\lambda_{\max}(\AA_t)\le 1- \frac{\eps}{\log n}$.
\end{problem} 
We defer the proof of the standard reduction stated below to the appendix.
\begin{restatable}{lemma}{Decision}\label{lem:Decision}
Given an algorithm $\mathcal{A}$ that solves the decision problem {\sc DecMaxEV}($\epsilon,\AA_0,\vv_1,\cdots,\vv_T$) (Definition~\ref{def:DecMaxEval}) for any $\epsilon>0$, $\AA_0 \succeq 0$ and vectors $\vv_1,\cdots,\vv_T$ in time $\mathcal{T}$, we can solve Problem~\ref{prob:dyn} in total time $O\left(\frac{\log^2 n\log\frac{n}{\epsilon}}{\epsilon}\cdot nnz(\AA_0) + \frac{\log n}{\epsilon}\log \frac{\lambda_{\max}(\AA_0)}{\lambda_{\max}(\AA_T)}\mathcal{T}\right)$.
\end{restatable}
Next, we describe \Cref{alg:PowerMethod} which can be viewed as an algorithm for \Cref{def:DecMaxEval} when there are no updates. 
Our algorithm essentially applies the {\it power iteration}, which is a standard algorithm used to find an approximate maximum eigenvalue and eigenvector of a matrix. In the algorithm, we make $R = O(\log n)$ copies to boost the probability.
\begin{algorithm}
\caption{{\sc DecMaxEV} with no update}\label{alg:PowerMethod}
 \begin{algorithmic}[1]
\Procedure{PowerMethod}{$\epsilon, \AA$}
\State $R \leftarrow 10\log n, r_0 \leftarrow 1$
\State $K \leftarrow \frac{4\log \frac{n}{\epsilon}}{\epsilon}$
\For{$r = 1,\cdots , R$}
\State $\vv^{(0,r)} \leftarrow $ random vector with coordinates chosen from $N(0,1)$\label{algline:randomInit}
\For{$k = 1:K$}
\State $\vv^{(k,r)}\leftarrow \AA\vv^{(k-1,r)}$
\EndFor
\State $\ww^{(r)} \leftarrow \frac{\vv^{(K,r)}}{\|\vv^{(K,r)}\|}$ 
\EndFor
\State $\WW = [\ww^{(1)},\dots,\ww^{(R)}]$\label{line:before case}
\If{ $(\ww^{(r)})^{\top}\AA\ww^{(r)} < 1-\epsilon$ for all $r\le R$}\label{line: if all w 1-eps}
\State \Return {\sc False}
\Else
\State $r_0\leftarrow$ smallest $r$ such that $(\ww^{(r)})^{\top}\AA\ww^{(r)} \geq 1-5\epsilon$
\State \Return $[r_0,\WW]$
\EndIf
\EndProcedure 
 \end{algorithmic}
\end{algorithm}
Below, we state the guarantees of the power method.

\begin{restatable}{lemma}{PowerMethod}\label{thm:StaticPowerMain}
Let $\epsilon>0$ and $\AA \succeq 0$. Let $\WW$ be as defined in Line~\ref{line:before case} in the execution of {\sc PowerMethod}($\epsilon,\AA$). With probability at least $1-1/n^{10}$, for some $\ww \in \WW$, it holds that $\ww^{\top}\AA\ww \geq (1-\frac{\epsilon}{2})\lambda_{\max}(\AA)$. The total time taken by the algorithm is at most $O\left(\frac{nnz(\AA)\log n\log \frac{n}{\epsilon}}{\epsilon}\right)$.

Furthermore, let $\lambda_i$ and $\uu_i$ denote the eigenvalues and eigenvectors of $\AA$. For all $i$ such that $\lambda_i (\AA) \leq \frac{\lambda_{\max}(\AA)}{2}$, with probability at least $1-2/n^{10}$, $\left[\ww^{\top}\uu_i\right]^2\leq \frac{1}{n^8}\cdot \frac{\lambda_i}{\lambda_1}$.
\end{restatable}
 We note that the last line of the above lemma is saying that the vectors returned by the power method satisfy \Cref{def:super}, which we state for completeness but is not required by our algorithm. The following result is a direct consequence of Lemma~\ref{thm:StaticPowerMain}.

 \begin{corollary}\label{thm:StaticPower}
 Let $\epsilon>0, \AA \succeq 0$. Let $\WW$ be as defined in Line~\ref{line:before case} in the execution of {\sc PowerMethod}($\epsilon,\AA$). If $\lambda_{\max}(\AA) \ge 1-\epsilon$, then with probability at least $1-1/n^{10}$, $\ww^{\top}\AA\ww\geq 1-5\epsilon$ for some $\ww\in \WW$. Furthermore, if $\lambda_{\max}(\AA) \ge 1-\epsilon/\log n$, then with probability at least $1-1/n^{10}$, $\ww^{\top}\AA\ww\geq 1-\epsilon$ for some $\ww\in \WW$. The total time taken by the algorithm is at most $O\left(\frac{nnz(\AA)\log n\log \frac{n}{\epsilon}}{\epsilon}\right)$.
 \end{corollary}

Observe that, if the algorithm returns $[r_0,\WW]$, then $(\ww^{(r)})^{\top}\AA\ww^{(r)}\geq 1-5\epsilon$ for $r=r_0$, and $\ww^{(r_0)}$ is therefore a solution to Problem~\ref{def:DecMaxEval} when there is no update. The power method and its analysis are standard, and we thus defer the proof of \Cref{thm:StaticPowerMain} to the appendix.

Next, in \Cref{alg:Init,alg:DynamicMaxPM} we describe an algorithm for Problem~\ref{def:DecMaxEval} when we have an online sequence of updates $\vv_1,\cdots, \vv_T$. 
The algorithm starts by initializing $R = O(\log n)$ copies of the approximate maximum eigenvectors from the power method. Given a sequence of updates, as long as one of the copies is the witness that the current matrix $\AA_t$ still has a large eigenvalue, i.e., there exists $r$ where $(\ww^{(r)})_t^{\top}\AA_t\ww^{(r)}_t\geq 1-40\epsilon$, we can just return $\ww^{(r)}$ as the solution to Problem~\ref{def:DecMaxEval}. 
Otherwise, $(\ww^{(r)})_t^{\top}\AA_t\ww^{(r)}_t< 1-40\epsilon$ for all $r \le R$ and none of the vectors from the previous call to the power method are a witness of large eigenvalues anymore. In this case, we simply recompute these vectors by calling the power method again. If the power method returns that there is no large eigenvector, then we return {\sc False} from now. Otherwise, we continue in the same manner. 
Note that our algorithm is very simple, but as we will see, the analysis is not straightforward.
\begin{algorithm}
\caption{Initialization}\label{alg:Init}
 \begin{algorithmic}[1]
 \Procedure{Init}{$\epsilon, \AA_{0}$}
\State $\WW \leftarrow$ {\sc PowerMethod}($\epsilon,\AA_0$)\label{algline:PMInit}
\State \Return $\WW$
\EndProcedure
\end{algorithmic}
\end{algorithm}

\begin{algorithm}
\caption{Update algorithm at time $t$ ($A_{t-1},r_t,\WW_{t-1}= [w^{(r)}_{t-1}: r= 1,\cdots R], \eps$ are maintained)}\label{alg:DynamicMaxPM}
 \begin{algorithmic}[1]
\Procedure{Update}{$\vv_t$}
\State $\AA_t \leftarrow \AA_{t-1}-\vv_t\vv_t^{\top}$
\If{$ (\ww^{(r)}_{t-1})^{\top}\AA_t\ww^{(r)}_{t-1} <1-40\epsilon$ for all $r \le R$}\label{algline:Check}
\State $[r_t,\WW_t]\leftarrow$ {\sc PowerMethod}($\epsilon,\AA_t$)\label{algline:PM}
\If{{\sc PowerMethod}($\epsilon,\AA_t$) returns {\sc False}}
\State \Return {\sc False} for all further updates
\EndIf
\Else
\State $r_t\leftarrow $ smallest $r$ such that $(\ww^{(r)}_{t-1})^{\top}\AA_t\ww^{(r)}_{t-1} \geq 1-40\epsilon$
\State $\WW_t \leftarrow \WW_{t-1}$
\EndIf
\State \Return $[r_t,\WW_t]$
\EndProcedure 
 \end{algorithmic}
\end{algorithm}

\subsection{Proof Overview}
The overall proof of \Cref{thm:upper}, including the proof of correctness and the runtime depends on the number of executions in Line~\ref{algline:PM} in Algorithm~\ref{alg:DynamicMaxPM}. If the number of executions of Line~\ref{algline:PM} is bounded by $\poly(\log n/\epsilon)$, then the remaining analysis is straightforward. Therefore, the majority of our analysis is dedicated to proving this key lemma, i.e., $\poly(\log n/\epsilon)$ bound on the number of calls to the power method:
\begin{lemma}[Key Lemma]\label{lem:BoundW}
The number of executions of Line~\ref{algline:PM} over all updates is bounded by $O(\log n\log^5\frac{n}{\epsilon}/\epsilon^2)$ with probability at least $1-\frac{1}{n}$.
\end{lemma}
Given the key lemma, the correctness and runtime analyses are quite straightforward and are presented in \Cref{sec:correct}. We now give an overview of the proof of \Cref{lem:BoundW}.

Let us consider what happens between two consecutive calls to Line~\ref{algline:PM}, say at $\AA$ and $\AAtil = \AA - \sum_{i=1}^k\vv_i\vv_i^{\top}$. We first define the following subspaces of $\AA$ and $\AAtil$.
Recall \Cref{def:subspace}, which we use to define the following subspaces.

\begin{definition}[Subspaces of $\AA$ and $\AAtil$]\label{def:SpaceA} Given $\epsilon>0$, $\AA$, and $\AAtil$ define for $\nu = 0,1,\cdots, 15\log \frac{n}{\epsilon}-1$:
\[
T_{\nu} = \Span\left(\frac{(\nu+1)\epsilon}{5\log \frac{n}{\epsilon}},\AA\right) -\Span\left(\frac{\nu\epsilon}{5\log \frac{n}{\epsilon}},\AA\right),
\]
and,
\[
\tilde{T}_{\nu} = \Span\left(\frac{(\nu+1)\epsilon}{5\log \frac{n}{\epsilon}},\AAtil\right)-\Span\left(\frac{\nu\epsilon}{5\log \frac{n}{\epsilon}},\AAtil\right).
\]
That is, the space $T_{\nu}$ and $\Ttil_\nu$ are spanned by eigenvectors of $\AA$ and $\AAtil$, respectively, corresponding to eigenvalues between $\left(1-(\nu+1)\frac{\epsilon}{5\log \frac{n}{\epsilon}}\right)\lambda_0$ and $\left(1-\nu\frac{\epsilon}{5\log \frac{n}{\epsilon}}\right)\lambda_0$.

Let $d_{\nu} = \dim{T_{\nu}}$ and $\tilde{d}_{\nu} = \dim{\tilde{T}_{\nu}}$. Also define,
\[
\tilde{T} = \Span(3\epsilon,\AAtil),\quad  T = \Span(3\epsilon,\AA),
\]
and let $d = \dim{T}$, $\tilde{d} = \dim{\tilde{T}}$.
\end{definition}

Observe that $T = \sum_{\nu = 0}^{15\log\frac{n}{\eps}-1}T_\nu$ and similarly $\Ttil = \sum_{\nu = 0}^{15\log\frac{n}{\eps}-1}\Ttil_\nu$. We next define some indices/levels corresponding to large subspaces, which we call ``important levels''.
\begin{definition}[Important $\nu$]\label{def:impNu} We say a level $\nu$ is important if,
\[d_{\nu} \geq \frac{\epsilon}{600\log^3 \frac{n}{\epsilon}} \sum_{\nu'<\nu}d_{\nu'}.
\]
We will use $\mathcal{I}$ to denote the set of $\nu$ that are important.
\end{definition}

The main technical lemma that implies \Cref{lem:BoundW} is the following:

\begin{restatable}[Measure of Progress]{lemma}{progress}\label{lem:EigenspaceChange}
	Let $\epsilon>0$ and let $\WW=[\ww^{(1)},\dots,\ww^{(R)}]$ be as defined in Line~\ref{line:before case} in the execution of {\sc PowerMethod}($\epsilon,\AA$). Let $\vv_{1},\cdots,\vv_{k}$ be a sequence of updates generated by an oblivious adversary and define $\AAtil=\AA-\sum_{i=1}^{k}\vv_{i}\vv_{i}^{\top}$.
	
	Suppose that $\lambda_{\max}(\AA)\ge 1-\epsilon$ and $\ww^{\top}\AAtil\ww<1-40\epsilon$ for all $\ww\in\WW$. Then, with probability at least $1-\frac{50\log\frac{n}{\epsilon}}{n^{2}}$, for some $\nu\in\mathcal{I}$,
	\begin{itemize}
		\item $\dim{T_{\nu}-\tilde{T}}\geq\frac{\epsilon}{300\log\frac{n}{\epsilon}}d_{\nu}$ if $d_{\nu}\geq\frac{3000\log n\log\frac{n}{\epsilon}}{\epsilon}$, or
		\item $\dim{T_{\nu}-\tilde{T}}\geq1$ if $d_{\nu}<\frac{3000\log n\log\frac{n}{\epsilon}}{\epsilon}$.
	\end{itemize}
\end{restatable}

We prove this lemma in \Cref{sec:progress}. Intuitively speaking, it means that, whenever  Line~\ref{algline:PM} of \Cref{alg:DynamicMaxPM} is executed, there is some important level $\nu$ such that an $\Omega(\eps/\poly \log(n/\eps))$-fraction of eigenvalues of $\AA$ at level $\nu$ have decreased in value. This is the crucial place where we exploit an oblivious adversary.

Given \Cref{lem:EigenspaceChange}, the remaining proof of \Cref{lem:BoundW} follows a potential function analysis which is presented in detail in Section~\ref{sec:proof key}. We consider potentials $\Phi_j = \sum_{\nu=0}^jd_{\nu}$ for $j = 0,\cdots,15\log\frac{n}{\epsilon}-1$. The main observation is that $\Phi_j$ is non-increasing over time for all $j$, and whenever there exists $\nu_0\in \mathcal{I}$ that satisfies the condition of Lemma~\ref{lem:EigenspaceChange}, $\Phi_{\nu_0}$ decreases by $\dim{T_{\nu_0}-\tilde{T}}$. Since $\dim{T_{\nu_0}-\tilde{T}} \geq \Omega(\epsilon/\poly\log(n/\epsilon))d_{\nu_0}$ and $\nu_0\in \mathcal{I}$, i.e., $\Phi_{\nu_0} = \sum_{\nu<\nu_0}d_{\nu} + d_{\nu_0} \leq d_{\nu_0} \left(\frac{O(\log^3\frac{n}{\epsilon})}{\epsilon}+1\right) $,  we can prove that $\Phi_{\nu_0}$ decreases by a multiplicative factor of $\Omega(1-\epsilon^2/\poly\log(n/\epsilon))$. As a result, every time our algorithm executes Line~\ref{algline:PM}, $\Phi_{j}$ decreases by a multiplicative factor for some $j$, and since we have at most $15\log\frac{n}{\epsilon}$ values of $j$, we can only have $\poly(\log n/\epsilon)$ executions of Line~\ref{algline:PM}.

It remains to describe how we prove Lemma~\ref{lem:EigenspaceChange} at a high level. We can write $\ww^{\top}\AAtil\ww$ for any $\ww \in \WW$ as
\[
\ww^{\top}\AAtil\ww = \ww^{\top}\AA\ww - \ww^{\top}\VV\ww,
\]
for $\VV = \sum_{i=1}^k\vv_i\vv_i^{\top}$. 
Our strategy is to show that:
\begin{align*}\label{star}
        &\text{If } \dim{T_{\nu}-\tilde{T}}\text{ does not satisfies the inequalities in \Cref{lem:EigenspaceChange} for all }\nu\in \mathcal{I},\\ 
        &\text{then }\ww^{\top}\VV\ww \leq 35\epsilon \text{ for all } \ww \in \WW. \tag{$\star$}
\end{align*}    
Given \eqref{star} as formalized later in \Cref{cl:progress}, we can conclude \Cref{lem:EigenspaceChange} because, from the definition of $\AA$ and $\AAtil$, we have that for some $\ww\in \WW$, $\ww^{\top}\AA\ww \geq 1-5\epsilon$ by \Cref{thm:StaticPower} and $\ww^{\top}\AAtil\ww <1-40\epsilon$. As a result for this $\ww$, $\ww^{\top}\VV\ww >35\epsilon$. Now, by contra-position of \eqref{star}, we have that $\dim{T_{\nu}-\tilde{T}}$ is large for some $\nu \in \mathcal{I}$.

To prove \eqref{star}, we further decompose $\ww^{\top}\VV\ww$ as
\[
\ww^{\top}\VV\ww = \ww^{\top}\VV_{\tilde{T}}\ww+\sum_{\nu = 0}^{15\log\frac{n}{\epsilon}-1}\ww^{\top}\VV_{T_{\nu} -\tilde{T}}\ww+ \ww^{\top}\VV_{\overline{T}}\ww.
\]

In the above equation, $\VV_{\tilde{T}}= \Pi_{\tilde{T}}\VV\Pi_{\tilde{T}},\VV_{T_{\nu} -\tilde{T}}=\Pi_{\nu}\VV\Pi_{\nu}$, and $\VV_{\overline{T}} = \Pi_{\overline{T}}\VV\Pi_{\overline{T}}$ where $\Pi_{\tilde{T}},\Pi_{\nu},\Pi_{\overline{T}}$ denote the projections matrices that project any vector onto the spaces $\tilde{T}$, $T_{\nu}-\tilde{T}$, and $\overline{T}$ respectively\footnote{Suppose a subspace $S$ is spanned by vectors $\uu_1,\dots,\uu_k$. Let $\UU = [\uu_1,\dots,\uu_k]$. Recall that the projection matrix onto $S$ is $\UU(\UU^\top \UU)^{-1} \UU^\top$.}. Refer to Section~\ref{sec:progress} for proof of why such a split is possible. Our proof of \eqref{star} then bounds the terms on the right-hand side. Let us consider each term separately.
\begin{enumerate}
	\item $\ww^{\top}\VV_{\tilde{T}}\ww$: We prove that this is always at most $10\epsilon(1+\epsilon)$ (Equation~\eqref{eq:V2}). From the definition of $\VV,$ 
	\[
	\ww^{\top}\VV_{\tilde{T}}\ww = \ww^{\top}\Pi_{\tilde{T}}\AA\Pi_{\tilde{T}}\ww - \ww^{\top}\Pi_{\tilde{T}}\AAtil\Pi_{\tilde{T}}\ww.
	\] 
	Since $\Pi_{\tilde{T}}\ww$ is the projection of $\ww$ along the large eigenspace of $\AAtil$, the second term on the right-hand side above is large, i.e. $\geq (1- 10\epsilon)\lambda_0\|\Pi_{\tilde{T}}\ww\|^2$. The first term on the right-hand side can be bounded as, $\ww^{\top}\Pi_{\tilde{T}}\AA\Pi_{\tilde{T}}\ww \leq \|\AA\|\|\Pi_{\tilde{T}}\ww\|^2 \leq \lambda_0 \|\Pi_{\tilde{T}}\ww\|^2$.
 Therefore the difference on the right-hand side is at most $10\epsilon\lambda_0\|\Pi_{\tilde{T}}\ww\|^2 \leq 10 \epsilon\lambda_0\|\ww\|^2 = 10\epsilon \lambda_0 \leq 10\epsilon(1+\epsilon)$.
	\item $\ww^{\top}\VV_{\overline{T}}\ww$: Observe that this term corresponds to the projection of $\ww$ along the space spanned by the eigenvalues of $\AA$ of size at most $1-3\epsilon$. Let $\uu_i$ and $\lambda_i$ denote an eigenvector and eigenvalue pair with $\lambda_i<1-3\epsilon$. Since the power method can guarantee that $\ww^{\top}{\uu_i}\approx\lambda_i^{2K}$, we have  $\lambda_i^{2K} \leq (1-3\epsilon)^{2K}\leq \poly\left(\frac{\epsilon}{n}\right)$ is tiny. So we have that $\ww^{\top}\VV_{\overline{T}}\ww\leq \epsilon$ (Lemma~\ref{lem:boundLowEV}).

	\paragraph{}Before we look at the final case, we define a basis for the space $T_{\nu}$.
	\begin{definition}[Basis for $T_{\nu}$]\label{def:Basis} Let $T_{\nu}$ be as defined in Definition~\ref{def:SpaceA}.  Define indices $a_{\nu}$ and $b_{\nu}$ with $b_{\nu}-a_{\nu}+1 = d_{\nu}$ such that the basis of $T_{\nu}$ is given by $\uu_{a_{\nu}},\cdots, \uu_{b_{\nu}}$, where $\uu_1,\uu_2,\cdots,\uu_n$ are the eigenvectors of $\AA$ in decreasing order of eigenvalues.
	\end{definition}

    \item $\ww^{\top}\VV_{T_{\nu} -\tilde{T}}\ww$: For this discussion, we will ignore the constant factors and assume that the high probability events hold. Let $\Pi_{\nu}$ denote the projection matrix for the space $T_{\nu}-\Ttil$. Observe that $\ww^{\top}\VV_{T_{\nu}-\Ttil}\ww=\ww^{\top}\Pi_{\nu}\VV\Pi_{\nu}\ww\le\|\VV\|\|\Pi_{\nu}\ww\|^2 \le (1+\epsilon)\|\Pi_{\nu}\ww\|^2$, where the last inequality is because $\AAtil=\AA-\VV\succeq0$, and therefore, $\|\VV\|\le\|\AA\|\le(1+\epsilon)$. Hence, it suffices to bound $\|\Pi_{\nu}\ww\|^2 = O(\epsilon)$. 

We can write $\ensuremath{\ww=\frac{\sum_{i=1}^{n}\lambda_{i}^{K}\alpha_{i}\uu_{i}}{\sqrt{\sum_{i=1}^{n}\lambda_{i}^{2K}\alpha_{i}^{2}}}}$ where $\lambda_{i},\uu_{i}$'s are the eigenvalues and eigenvectors of $\AA$ and $\alpha_{i}\sim N(0,1)$. 
Define $\zz = \sum_{i=1}^n z_i \uu_i$ where $z_i = \lambda_i^K\alpha_i$. That is, 
$\ww=\frac{\zz}{\|\zz\|}$. Since $\|\Pi_{\nu}\ww\| = \|\Pi_{\nu}\zz\|/\|\zz\|$, it suffices to show that $\|\Pi_{\nu}\zz\|^2 \le O(\epsilon) \| \zz\|^2$. We show this in two separate cases. In both cases, we start with the following bound
\[
\|\Pi_{\nu}\zz\|^{2}\le\lambda_{a_{\nu}}^{2K}\cdot\dim{T_{\nu}-\tilde{T}},
\]
which holds with high probability. To see this, let $\boldsymbol{g}_{\nu}\sim N(0,1)$ be a gaussian vector in the space $T_{\nu}-\tilde{T}$. 
We can couple $\textbf{g}_\nu$ with $\Pi_{\nu}\zz$ so that $\Pi_{\nu}\zz$ is dominated by $\lambda_{a_{\nu}}^{K}\cdot \boldsymbol{g}_{\nu}$. So $\|\Pi_{\nu}\zz\|^{2}\le\lambda_{a_{\nu}}^{2K}\|\boldsymbol{g}_{\nu}\|^{2}$. By \Cref{lem:NormG}, the norm square of gaussian vector is concentrated to its dimension so $\|\boldsymbol{g}_{\nu}\|^{2}\le\dim{T_{\nu}-\tilde{T}}$ with high probability, thus proving the inequality. Next, we will bound $\dim{T_{\nu}-\tilde{T}}$ in terms of $\|\zz\|$ in two cases. 

\paragraph{When $\nu\protect\notin{\cal I}$ (\Cref{lem:NotImp}):}

From the definition of the important levels, we have 
\[
\dim{T_{\nu}-\tilde{T}}\leq d_{\nu}\leq\frac{O(\epsilon)}{\log^{3}\frac{n}{\epsilon}}\sum_{\nu'<\nu}d_{\nu'}.
\]
Now, we have $\sum_{\nu'<\nu}d_{\nu'}\approx\sum_{i=1}^{b_{\nu-1}}\alpha_{i}^{2}$ because $\alpha_{i}\sim N(0,1)$ is gaussian and the norm square of gaussian vector is concentrated to its dimension (\Cref{lem:NormG}). Since $\alpha_{i}=z_{i}/\lambda_{i}^{K}$, we have that
\[
\sum_{\nu'<\nu}d_{\nu'}\approx\sum_{i=1}^{b_{\nu-1}}\alpha_{i}^{2}=\sum_{i=1}^{b_{\nu-1}}\frac{z_{i}^{2}}{\lambda_{i}^{2K}}\le\|\zz\|^{2}/\lambda_{b_{\nu-1}}^{2K}.
\]
Therefore, we have 
\[
\|\Pi_{\nu}\zz\|^2 \le\lambda_{a_{\nu}}^{2K}\dim{T_{\nu}-\tilde{T}}\le\left(\frac{\lambda_{a_{\nu}}}{\lambda_{_{b_{\nu-1}}}}\right)^{2K}\frac{O(\epsilon)}{\log^{3}\frac{n}{\epsilon}}\|\zz\|^{2}\le O(\epsilon)\|\zz\|^{2}
\]
where the last inequality is trivial because $\lambda_{b_{\nu-1}}\ge\lambda_{a_{\nu}}$ by definition.

\paragraph{When $\nu\in{\cal I}$ (\Cref{lem:GaussianProjD}):}

In this case, according to \eqref{star}, we can assume $$\dim{T_{\nu}-\tilde{T}}\lesssim\epsilon d_{\nu}.$$
Again, by \Cref{lem:NormG}, we have that $d_{\nu}\approx\sum_{i=a_{\nu}}^{b_{\nu}}\alpha_{i}^{2}$ because $\alpha_{i}\sim N(0,1)$ is gaussian. Since $\alpha_{i}=z_{i}/\lambda_{i}^{K}$, we have
\[
d_{\nu}\approx\sum_{i=a_{\nu}}^{b_{\nu}}\alpha_{i}^{2}=\sum_{i=a_{\nu}}^{b_{\nu}}\frac{z_{i}^{2}}{\lambda_{i}^{2K}}\le\|\zz\|^{2}/\lambda_{b_{\nu}}^{2K}.
\]
Therefore,
\[
\|\Pi_{\nu}\zz\|^2\le\lambda_{a_{\nu}}^{2K}\dim{T_{\nu}-\tilde{T}}\le\left(\frac{\lambda_{a_{\nu}}}{\lambda_{_{b_{\nu}}}}\right)^{2K}\epsilon\|\zz\|^{2}\le O(\epsilon)\|z\|^{2}
\]
where the last inequality is because $\ensuremath{\left(\frac{\lambda_{a_{\nu}}}{\lambda_{b_{\nu}}}\right)^{2K}\leq\left(\frac{1-\frac{\nu\epsilon}{5\log\frac{n}{\epsilon}}}{1-\frac{(\nu+1)\epsilon}{5\log\frac{n}{\epsilon}}}\right)^{2K}\approx\left(1+\frac{\epsilon}{2\log\frac{n}{\epsilon}}\right)^{2K}\approx O(1)}.$

\end{enumerate}
From these three cases, we can conclude that if $\dim{T_{\nu}-\tilde{T}}$ is small for all $\nu\in \mathcal{I}$, then $\ww^{\top}\VV\ww \leq 35\epsilon$, proving our claim.

In the remaining sections, 
we give formal proofs of the claims made in this section. In \Cref{sec:correct}, we prove the main result, \Cref{thm:upper}, assuming the key lemma. In \Cref{sec:proof key}, we prove the key lemma, \Cref{lem:BoundW}, assuming the \Cref{lem:EigenspaceChange}. Finally, we prove \Cref{lem:EigenspaceChange} in \Cref{sec:progress}

\subsection{Proof of the Main \Cref{thm:upper} assuming the Key \Cref{lem:BoundW}}\label{sec:correct}

Here, we formally prove \Cref{thm:upper} assuming the key \Cref{lem:BoundW}. We will first prove the correctness and then bound the total runtime.

\paragraph{Correctness.}
The following formalizes the correctness guarantee of  Algorithm~\ref{alg:DynamicMaxPM}. 

\begin{lemma}\label{lem:DynCorrectAns}
Let $\epsilon>1/n$. With probability at least $1-1/n$, the following holds for all time step $t \ge 1$: 
if the maximum eigenvalue of $\AA_t$ is at least $1-\frac{\epsilon}{\log n}$, {\sc Update}($\vv_t$) returns $[r_t,\WW_t]$.
\end{lemma}
\begin{proof}
We upper bound the probability that Algorithm {\sc Update}($\vv_t$) returns \textsc{False}. This can happen only when  Line~\ref{algline:PM} is executed and {\sc PowerMethod}($\epsilon,\AA_t$) returns {\sc False}, which happens with probability at most $1/n^{10}$ by \Cref{thm:StaticPower}.
By Lemma~\ref{lem:BoundW}, we execute Line~\ref{algline:PM} at most $O(\log n\log^5\frac{n}{\epsilon}/\epsilon^2)$ times throughout the whole update sequence. Therefore, by union bounds, {\sc Update}($\vv_t$) may return \textsc{False} with probability at most $\poly(\log(n)/\eps)/n^{10} \le 1/n$.
\end{proof}

\paragraph{Runtime.}\label{sec:runtime}

Next, we bound the runtime of the various lines of Algorithm~\ref{alg:DynamicMaxPM}. 
\begin{lemma}\label{lem:SameW}
For a fixed $\ww$ and any $t$, we can update $\ww^{\top}\AA_{t-1}\ww$ to $\ww^{\top}\AA_t\ww$ in time $O(nnz(\vv_t))$.
\end{lemma}
\begin{proof}
Note that,
\[
\ww^{\top}\AA_t\ww = \ww^{\top}\AA_{t-1}\ww - (\vv_t^{\top}\ww)^2.
\]
The second term can be computed in time $O(nnz(\vv_t))$.
\end{proof}
\begin{lemma}\label{lem:DiffW}
Fix time $t$. Given $\ww$ as input, we have that $\ww^{\top}\AA_t\ww$ and $\AA_t\ww$ can be computed $O\left(nnz(\AA_0) +\sum_{i=1}^t nnz(\vv_i)\right)$ time.
\end{lemma}
\begin{proof}
We can split $\AA_t$, and get,
\[
\ww^{\top}\AA_t\ww = \ww^{\top}\AA_0\ww - \sum_{i=1}^t (\vv_i^{\top}\ww)^2,\quad \text{and,}\quad
\AA_t\ww = \AA_0\ww - \sum_{i=1}^t \vv_i(\vv_i^{\top}\ww).
\]
The first term in both expressions can be computed in time $O(nnz(\AA_0))$ and for every $i$ we can compute $\vv_i^{\top}\ww$ in time $O(nnz(\vv_i))$, thus concluding the proof.
\end{proof}
\begin{lemma}\label{lem:PMt}For any time $t$, we can implement {\sc PowerMethod}($\epsilon,\AA_t$) in time at most 
\[
O\left(\frac{\log n\log \frac{n}{\epsilon}}{\epsilon}\left(nnz(\AA_0) +\sum_{i=1}^t nnz(\vv_i)\right)\right).
\]
\end{lemma}
\begin{proof}
When we call {\sc PowerMethod}($\epsilon,\AA_t$) from \Cref{alg:PowerMethod}, we can to compute $\AA_t\ww$ for some vector $\ww$ for $R\cdot K$ times and, at the end, compute $\ww \AA_t\ww$ for some vector $\ww$ for $R$ times. 
By \Cref{lem:DiffW}, this takes $O( (nnz(\AA_0) +\sum_{i=1}^t nnz(\vv_i) ) \cdot \log n\log (\frac{n}{\epsilon})/\epsilon)$ because $R = O(\log n)$ and $K = O(\log(\frac{n}{\epsilon})/\epsilon)$.

\end{proof}

 Given the above results, we can now prove Theorem~\ref{thm:upper}.
\subsubsection*{Proof of Theorem~\ref{thm:upper}}
\begin{proof}

We will prove that Algorithms~\ref{alg:Init} and \ref{alg:DynamicMaxPM} solve Problem~\ref{def:DecMaxEval}, which implies an algorithm for \Cref{prob:dyn} by \Cref{lem:Decision}.

From \Cref{thm:StaticPower}, Algorithm~\ref{alg:Init} solves Problem~\ref{def:DecMaxEval} for the initial matrix $\AA_0$ with probability at least $1-1/n$. From \Cref{lem:DynCorrectAns}, Algorithm~\ref{alg:DynamicMaxPM} returns $[r_t,\WW_t]$ for all $t\geq 1$, whenever the maximum eigenvalue of $\AA_t$ is at least $1-\epsilon/\log n$, with probability at least $1-1/n$. Note that, $\ww^{(r_t)}\in \WW_t$ is the required solution. Therefore, our algorithm returns the correct solution with a probability of at least $1-2/n$.

We now bound the total runtime of the algorithm over all time $t$. Define an indicator function $I_t$ as $I_t = 1$, if we execute Line~\ref{algline:PM} in {\sc Update}$(\vv_t)$ and $0$ otherwise. From Lemmas~\ref{lem:SameW} and \ref{lem:DiffW}, executing Line~\ref{algline:Check} at $t$ requires time $O(nnz(\vv_t)\log n)$ if $I_{t-1} = 0$ and time $O\left(\frac{\log n\log \frac{n}{\epsilon}}{\epsilon}\left(nnz(\AA_0) +\sum_{i=1}^t nnz(\vv_i)\right)\right)$ if $I_{t-1} = 1$. Summing these up, the total time $\mathcal{T}$ required for solving \Cref{def:DecMaxEval} is at most,
\begin{align*}
\mathcal{T} \leq &  O\left(\underbrace{\frac{\log n\log \frac{n}{\epsilon} \cdot nnz(\AA_0)}{\epsilon}}_{\text{Time for {\sc Init}($\epsilon,\AA_0$)}} + \sum_t I_t \underbrace{\frac{\log n\log \frac{n}{\epsilon}}{\epsilon}\left(nnz(\AA_0) +\sum_{i=1}^t nnz(\vv_i)\right)}_{\text{Time required at $t$ when $I_t = 1$}} + \sum_t (1-I_t)\underbrace{\log n \cdot nnz(\vv_t)}_{\text{Time when $I_t=0$}}\right)\\
\leq & O\left( \frac{\log n\log \frac{n}{\epsilon}\cdot nnz(\AA_0)}{\epsilon}+ \sum_t I_t \frac{\log n\log \frac{n}{\epsilon}}{\epsilon}\left(nnz(\AA_0) +\sum_{i=1}^t nnz(\vv_i)\right) + \sum_t\log n\cdot  nnz(\vv_t)\right)\\
\substack{(a)\\ \leq}& O\left( \frac{\log n\log^2 \frac{n}{\epsilon}\cdot nnz(\AA_0)}{\epsilon}+ \frac{\log n\log^5\frac{n}{\epsilon}}{\epsilon^2}\frac{\log n\log \frac{n}{\epsilon}}{\epsilon} \left(nnz(\AA_0) +\sum_{i=1}^T nnz(\vv_i)\right) + \sum_{t=1}^T\log n\cdot nnz(\vv_t)\right)\\
\substack{(b)\\ \leq} & O\left( \frac{\log^2 n\log^6\frac{n}{\epsilon}}{\epsilon^3} \left(nnz(\AA_0) +\sum_{i=1}^T nnz(\vv_i)\right) \right),
\end{align*}
with probability at least $1-\frac{1}{n}$. In $(a)$, we used Lemma~\ref{lem:BoundW} which states that the maximum number of $t$ for which $I_t=1$ is at most $O(\log n\log^5\frac{n}{\epsilon}/\epsilon^2)$ with probability at least $1-1/n$. In $(b)$, we combined all the terms. Therefore, by \Cref{lem:Decision}, our algorithms solve Problem~\ref{prob:dyn} with probability at least $1-3/n$ in total time
\begin{multline*}
O\left(\frac{\log^2 n\log\frac{n}{\epsilon}}{\epsilon}\cdot nnz(\AA_0) + \frac{\log n}{\epsilon}\log \frac{\lambda_{\max}(\AA_0)}{\lambda_{\max}(\AA_T)}\mathcal{T}\right) \leq \\ O\left(\frac{\log^3 n\log^6 \frac{n}{\epsilon}\log \frac{\lambda_{\max}(\AA_0)}{\lambda_{\max}(\AA_T)}}{\epsilon^4}\left(nnz(\AA_0) +\sum_{i=1}^T nnz(\vv_i)\right)\right).
\end{multline*}
\end{proof}


\subsection{Proof of the Key Lemma~\ref{lem:BoundW}: Few Executions of Line~\ref{algline:PM}}

\label{sec:proof key}

In this section, we prove Lemma~\ref{lem:BoundW} assuming the Progress Lemma~\ref{lem:EigenspaceChange}. 
Let us first recall the precise definition of $\AA$ and $\AAtil$. Suppose we execute Line~\ref{algline:PM} at update $t_{0}$. Now consider a sequence of updates, $\vv_{t_{0}+1},\cdots,\vv_{t_{0}+k}$, and let $\AA_{t_{0}+k}=\AA_{t_{0}}-\sum_{i=1}^{k}\vv_{t+i}\vv_{t+i}^{\top}.$ Suppose the next execution of Line~\ref{algline:PM} happens at $t_{0}+k$. For this to happen we must have that $\ww_{t_{0}}^{\top}\AA_{t_{0}+k}\ww_{t_{0}}<1-40\epsilon$ for all $\ww_{t_{0}}\in\WW_{t_{0}}$. We let $\AA=\AA_{t_{0}}$ and $\AAtil=\AA_{t_{0}+k}$. 

Next, recall Definition~\ref{def:SpaceA} and define $T_{\leq i}=\sum_{\nu=0}^{i}T_{\nu}$, $\tilde{T}_{\leq i}=\sum_{\nu=0}^{i}\tilde{T}_{\nu}$. The following observation will motivate our potential function analysis.
\begin{lemma}
	\label{lem:Monotone} For all $i\leq15\log\frac{n}{\epsilon}$, 
	\[
	\dim{T_{\le i}}\ge\dim{\Ttil_{\le i}}.
	\]
	Let $\nu_{0}$ be such that $\dim{T_{\nu_{0}}-\tilde{T}}>0$. For any $i\ge\nu_{0}$, we have
	\[
	\dim{T_{\le i}}\geq\dim{\Ttil_{\le i}}+\dim{T_{\nu_{0}}-\Ttil}.
	\]
	
\end{lemma}

\begin{proof}
	We claim that $\tilde{T}_{\leq i}\subseteq T_{\leq i}$, which implies the first claim. This is because the updates are decreasing. So, if $\vv^{\top}\AAtil\vv\geq1-\frac{(i+1)\epsilon}{5\log\frac{n}{\epsilon}}$ then $\vv^{\top}\AA\vv\geq1-\frac{(i+1)\epsilon}{5\log\frac{n}{\epsilon}}$. That is, if $\vv\in\Ttil_{\le i}$, then $\vv\in T_{\leq i}$. For the second part, we have $T_{\leq i}=T_{\leq i}\cap\tilde{T}_{\leq i}+(T_{\leq i}-\tilde{T}_{\leq i})\supseteq\tilde{T}_{\leq i}+(T_{\nu_{0}}-\tilde{T})$ because $T_{\nu_{0}}\subseteq T_{\le i}$ and $\Ttil_{\le i}\subseteq\Ttil$. Since $\tilde{T}_{\leq i}\cap(T_{\nu_{0}}-\tilde{T})=\emptyset$, we can conclude the second part.
\end{proof}

\paragraph{The Potentials.}
The above lemma and \Cref{lem:EigenspaceChange} motivate the following potentials.
For every $j\le15\log\frac{n}{\epsilon}$, we define the potentials $\Phi_{j}=\dim{T_{\le j}}$ and $\tilde{\Phi}_{j}=\dim{\Ttil_{\le j}}$. For all $j$, the potential may only decrease, i.e., $\Phitil_{j}\le\Phi_{j}$ by \Cref{lem:Monotone}. Also, clearly, $\Phi_{j}\le n$. 

We will show that for each execution of Line~\ref{algline:PM}, $\Phitil_{j}$ decreases from $\Phi_j$ by a significant factor with high probability for some $j$. This will bound the number of executions.

Consider any important level $\nu_{0} \in \cal{I}$. We have two observations:
\begin{enumerate}
	\item \label{enu:phi 1} $\Phitil_{\nu_{0}}\leq\Phi_{\nu_{0}}-\dim{T_{\nu_{0}}-\tilde{T}}$, and 
	\item \label{enu:phi 2} $d_{\nu_{0}}\geq\Omega(1)\frac{\epsilon}{\log^{3}\frac{n}{\epsilon}}\Phi_{v_{0}}$. 
\end{enumerate}
The first point follows directly from the second part of Lemma~\ref{lem:Monotone}. Moreover, since $\nu_{0}\in\mathcal{I}$ is important, we have $d_{\nu_{0}}\geq\frac{\epsilon}{600\log^{3}\frac{n}{\epsilon}}\sum_{\nu'<\nu}d_{\nu'}.$ Therefore, 
\[
\Phi_{\nu_{0}}=\sum_{\nu=1}^{\nu_{0}}d_{\nu}=\sum_{\nu'<\nu_{0}}d_{\nu'}+d_{\nu_{0}}\leq\left(\frac{600\log^{3}\frac{n}{\epsilon}}{\epsilon}+1\right)d_{\nu_{0}},
\]
implying the second point. 
Combining these observations with the Progress \Cref{lem:EigenspaceChange}, we have:
\begin{lemma}
	\label{claim:potential decrease}Suppose that $\lambda_{\max}(\AA)\ge1-\epsilon$ and $\ww^{\top}\AAtil\ww<1-40\epsilon$ for all $\ww\in\WW$. With probability at least $1-50\log\frac{n}{\epsilon}/n^{2}$, there is a level $\nu_{0}$ such that either 
	\begin{itemize}
		\item $\Phitil_{\nu_{0}}\le\left(1-\frac{\Omega(\epsilon^{2})}{\log^{4}\frac{n}{\epsilon}}\right)\Phi_{\nu_{0}}$, or 
		\item $\Phitil_{\nu_{0}}\le\Phi_{\nu_{0}}-1$ and $\Phi_{\nu_{0}}\leq O(1)\frac{\log^{4}\frac{n}{\epsilon}\log n}{\epsilon^{2}}.$
	\end{itemize}
\end{lemma}

\begin{proof}
	Given the assumption, it follows from \Cref{lem:EigenspaceChange} that with probability at least $1-50\log\frac{n}{\epsilon}/n^{2}$, there is an important level $\nu_{0}\in\mathcal{I}$ be such that either $d_{\nu_{0}}\geq\frac{3000\log n\log\frac{n}{\epsilon}}{\epsilon}$ and $\dimm(T_{\nu_{0}}-\tilde{T})\geq\frac{\epsilon}{300\log\frac{n}{\epsilon}}d_{\nu_{0}}$ or $d_{\nu_{0}}<\frac{3000\log n\log\frac{n}{\epsilon}}{\epsilon}$ and $\dimm(T_{\nu_{0}}-\tilde{T})\geq1$. So, by calculation, we have the following.
	\begin{itemize}
		\item If $d_{\nu_{0}}\geq\frac{3000\log n\log\frac{n}{\epsilon}}{\epsilon}$, we have $\Phitil_{\nu_{0}}\leq\Phi_{\nu_{0}}-\frac{\epsilon}{300\log\frac{n}{\epsilon}}d_{\nu_{0}}\le\Phi_{\nu_{0}}(1-\frac{\Omega(\epsilon^{2})}{\log^{4}\frac{n}{\epsilon}})$ where the first inequality is by (\ref{enu:phi 1}) and \Cref{lem:EigenspaceChange}. The second is by (\ref{enu:phi 2}). 
		\item If $d_{\nu_{0}}<\frac{3000\log n\log\frac{n}{\epsilon}}{\epsilon}$, we have $\Phitil_{\nu_{0}}\le\Phi_{\nu_{0}}-1$ by (\ref{enu:phi 1}) and \Cref{lem:EigenspaceChange}. In this case, we also have $\Phi_{\nu_{0}}\leq O(1)\frac{\log^{4}\frac{n}{\epsilon}\log n}{\epsilon^{2}}$ by (\ref{enu:phi 2}). 
	\end{itemize}
\end{proof}
We are now ready to prove Lemma~\ref{lem:BoundW}. 

\subsubsection*{Proof of Lemma~\ref{lem:BoundW}.}

First, observe that whenever $\lambda_{\max}(\AA)<1-\epsilon$, by Line~\ref{line: if all w 1-eps} of \Cref{alg:PowerMethod}, we will always return $\textsc{False}$ at the next execution of Line \ref{algline:PM} of \Cref{alg:DynamicMaxPM}.

Therefore, it suffices to bound the number of executions while $\lambda_{\max}(\AA)\ge1-\epsilon$. We execute Line~\ref{algline:PM} only if $\ww^{\top}\AAtil\ww<1-40\epsilon$ for all $\ww\in\WW$. When this happens, there exists a level $j$ where the potential $\Phi_{j}$ significantly decreases according to \Cref{claim:potential decrease} with probability at least $1-50\log\frac{n}{\epsilon}/n^{2}$. For each level $j$, this can happen at most $L=O(\frac{\log n\log^{4}\frac{n}{\epsilon}}{\epsilon^{2}})$ times because for every $j$, $\Phi_{j}$ is an integer that may only decrease and is bounded by $n$. Suppose for contradiction that there are more than $L\times15\log\frac{n}{\epsilon}$ executions of Line~\ref{algline:PM}. So, with probability at least $1-L\cdot50\log\frac{n}{\epsilon}/n^{2}\ge1/n$, there exists a level $j$ where $\Phi_{j}$ decreases according to \Cref{claim:potential decrease} strictly more than $L$ times. This is a contradiction.

\subsection{Proof of the Progress Lemma~\ref{lem:EigenspaceChange}}\label{sec:progress}
It remains to prove the Progress Lemma. We first restate the lemma here.

\progress*
Recall the definitions of subspaces $T,T_\nu,\Ttil$ and $\overline{T}$ in \Cref{def:SpaceA}.
We will first state the following claim and show that this is sufficient to prove \Cref{lem:EigenspaceChange}. After concluding the proof of \Cref{lem:EigenspaceChange}, we would prove the claim.

\begin{claim}\label{cl:progress} Suppose that $\lambda_{\max}(\AA) \geq 1-\epsilon$. If for every $\nu \in \mathcal{I}$, 
\begin{itemize}
\item $\dim{T_{\nu} -\tilde{T}} < \frac{\epsilon}{300\log\frac{n}{\epsilon}} d_{\nu}$ if $d_{\nu} \geq \frac{3000\log n\log\frac{n}{\epsilon}}{\epsilon}$, and 
\item $\dim{T_{\nu} -\tilde{T}} < 1$ if $d_{\nu} < \frac{3000\log n\log\frac{n}{\epsilon}}{\epsilon}$.
\end{itemize}
Then, with probability at least $1-\frac{40\log\frac{n}{\epsilon}}{n^2}$, $\ww^{\top}\VV\ww\leq 35 \epsilon$ for all $\ww\in \WW$.
\end{claim}
\subsubsection*{Proof of \Cref{lem:EigenspaceChange} using \Cref{cl:progress}}
Suppose for contradiction that \Cref{lem:EigenspaceChange} does not hold, i.e., the conditions on $\dim{T_\nu-\Ttil}$ of \Cref{cl:progress} hold for all $\nu \in \cal{I}$.
 On one hand, since $\lambda_{\max}(\AA)\geq 1-\epsilon$, \Cref{cl:progress} says that $\ww^{\top}\VV\ww\leq 35 \epsilon$ for all $\ww \in \WW$ with probability at least $1-\frac{40\log\frac{n}{\epsilon}}{n^2}$.
 From the assumption of \Cref{lem:EigenspaceChange}, we have $\ww^{\top}\AAtil\ww < 1- 40\epsilon$ for all $\ww\in \WW$ as well, this implies that, for all $\ww \in \WW$,
\[
\ww^{\top}\AA\ww < \ww^{\top}\AAtil\ww +\ww^{\top}\VV\ww < 1-5\epsilon.
\]
On the other hand, since $\lambda_{\max}(\AA)\geq 1- \epsilon$, we must have  $\ww^{\top}\AA\ww \geq 1-5\epsilon$ for some $\ww\in \WW$ with probability at least $1-1/n^2$ by  \Cref{thm:StaticPower}.
This gives a contradiction. 

Therefore, we conclude that, with probability at least $1-\frac{40\log\frac{n}{\epsilon}+1}{n^2}$, the conditions of \Cref{cl:progress} must be false for some $\nu \in \mathcal{I}$. That is, we have
\begin{itemize}
\item $\dim{T_{\nu} -\tilde{T}} \geq \frac{\epsilon}{300\log\frac{n}{\epsilon}} d_{\nu}$ if $d_{\nu} \geq \frac{3000\log n\log\frac{n}{\epsilon}}{\epsilon}$, and 
\item $\dim{T_{\nu} -\tilde{T}} \geq 1$ if $d_{\nu} < \frac{3000\log n\log\frac{n}{\epsilon}}{\epsilon}$.
\end{itemize}
This concludes the proof of \Cref{lem:EigenspaceChange}.

\subsubsection*{Setting Up for the Proof of \Cref{cl:progress}}

Recall the definitions of subspaces $T,T_\nu,\Ttil$ and $\overline{T}$ in \Cref{def:SpaceA}.

\begin{proposition}\label{prop:decomp space}
We can cover the entire space with the following subspaces
\begin{equation}\label{eq:SplitSpace}
 \mathbb{R}^n = T + \overline{T} = \tilde{T} + \sum_{\nu = 0}^{15\log \frac{n}{\epsilon}-1}(T_{\nu} -\tilde{T}) + \overline{T}
\end{equation}
where all subspaces in the sum are mutually orthogonal.
\end{proposition}
\begin{proof}
    It suffices to show that $T = \tilde{T} + \sum_{\nu = 0}^{15\log \frac{n}{\epsilon}-1}(T_{\nu} -\tilde{T})$.
    To see this, note that $\tilde{T}\subseteq T$. Therefore, we have $T = (T-\tilde{T}) + T\cap \tilde{T} = (T - \tilde{T}) + \tilde{T}$. We also know that $ T = \sum_{\nu=0}^{15\log \frac{n}{\epsilon}-1}T_{\nu}$ and this gives
$T - \tilde{T} = \sum_{\nu=0}^{15\log \frac{n}{\epsilon}-1} (T_{\nu}-\tilde{T})$, which concludes the proof.
\end{proof}

\paragraph{Notation.}The goal of \Cref{cl:progress} is to bound $\ww^{\top}\VV\ww\leq 35 \epsilon$ for all $\ww\in \WW$. We will use the following notations.
\begin{itemize}
    \item Let $\Pi_{\tilde{T}},\Pi_{\nu},\Pi_{\overline{T}}$ denote projection matrices to the subspaces $\tilde{T}$, $T_{\nu}-\tilde{T}$, and $\overline{T}$ respectively.
    \item  Define $\VV_{\tilde{T}} = \Pi_{\tilde{T}}\VV\Pi_{\tilde{T}}$, $\VV_{T_{\nu}-\tilde{T}} = \Pi_{\nu}\VV\Pi_{\nu}$, and $\VV_{\overline{T}} = \Pi_{\overline{T}}\VV\Pi_{\overline{T}}$.
\end{itemize}
By \Cref{prop:decomp space}, for any $\ww\in \WW$, we can decompose $\ww^{\top}\VV\ww$  as 
\[
\ww^{\top}\VV\ww = \ww^{\top}\VV_{\tilde{T}}\ww+\sum_{\nu = 0}^{15\log\frac{n}{\epsilon}-1}\ww^{\top}\VV_{T_{\nu} -\tilde{T}}\ww+ \ww^{\top}\VV_{\overline{T}}\ww.
\]
Our strategy is to upper bound each term one by one. Bounding $\ww^{\top}\VV_{\tilde{T}}\ww$ is straightforward, but bounding other terms requires technical helper lemmas.
\Cref{lem:boundLowEV} is needed for bounding $\ww^{\top}\VV_{\overline{T}}\ww$.
\Cref{lem:GaussianProjD,lem:NotImp} are helpful for bounding $\ww^{\top}\VV_{T_{\nu} -\tilde{T}}\ww$ when $\nu \in \cal{I}$ and when $\nu \notin \cal{I}$, respectively.

\subsubsection*{Helper Lemmas for \Cref{cl:progress}}
In all the statements of the helper lemmas below. Let $\WW$ be as defined in Line~\ref{line:before case} in the execution of \textsc{PowerMethod}($\epsilon,\AA$). Consider any fixed $\ww\in\WW$. Observe that we can write 
\begin{equation}\label{eq:rewrite w}  
\ww=\sum_{i=1}^{n}\frac{\lambda_{i}^{K}\alpha_{i}\uu_{i}}{\sqrt{\sum_{j}\lambda_{j}^{2K}\alpha_{j}^{2}}}
\end{equation}
where $K=\frac{4\log\frac{n}{\epsilon}}{\epsilon}$, $\alpha_{i}\sim N(0,1)$ are gaussian random variables, and $\lambda_i$ and $\uu_{i}$ are the $i$-th eigenvalue and eigenvector of $\AA$, respectively.

The following lemma shows that the projection of $\ww$ on $\overline{T}$ is always small. At a high level, 
since $\overline{T}$ is spanned by the eigenvectors with the small eigenvalues, the power method guarantees with high probability that the direction of $\ww$ along these eigenvectors will be exponentially small in the number of iterations $K$.
Recall that $\Pi_{\overline{T}}$ is the projection matrix to the space $\overline{T}$.
\begin{lemma}\label{lem:boundLowEV}
If $\lambda_{\max}(\AA)\geq 1-\epsilon$, then
\[
\Pr\left[\|\Pi_{\overline{T}}\ww\|^2 \leq \frac{\epsilon^2}{4}\right]\geq 1 - \frac{3}{n^2}.
\]
\end{lemma}
\begin{proof}
From the definition of $\overline{T}$ and $d$ defined in \Cref{def:SpaceA}, we have $\Pi_{\overline{T}}\ww = \frac{\sum_{i>d} \lambda_i^{K}\alpha_i\uu_i}{\sqrt{\sum_{j=1}^n \lambda_j^{2K}\alpha_i^2}}$ by \Cref{eq:rewrite w}. Hence,  
\[
\|\Pi_{\overline{T}}\ww\|^2 = \frac{\sum_{i>d}\lambda_i^{2K}\alpha_i^2}{\sum_{i=1}^n\lambda_i^{2K}\alpha_i^2}.
\]
First, we give a crude lower bound for the denominator. We have $$\sum_{i=1}^n\lambda_i^{2K}\alpha_i^2 \geq \frac{\lambda_1^{2K}}{n^4}$$ with probability at least $1-1/n^{2}$. This is because $\sum_{i=1}^n\lambda_i^{2K}\alpha_i^2 \geq \lambda_1^{2K}\alpha_1^2$ and,  since $\alpha_1^2 \sim \chi^2_1$,  we have $\alpha_1^2 \geq 1/n^4$ with probability at least $1-1/n^{2}$ by \Cref{lem:chi}.

Next, we upper bound the numerator as 
\[
\sum_{i>d}\lambda_i^{2K}\alpha_i^2 \leq \lambda_{d+1}^{2K}\sum_{i=1}^n\alpha_i^2.
\]
From \Cref{lem:NormG}, $\sum_{i=1}^n\alpha_i^2 \leq 2 n$ with probability at least $1-1/n^2$. Also note that, since $\lambda_{\max}(\AA)\geq 1-\epsilon \geq \lambda_0 (1-2\epsilon)$. We now have with probability $1-3/n^2$,
\[
\|\Pi\ww\|^2 \leq \left(\frac{\lambda_{d+1}}{\lambda_1}\right)^{2K}\cdot 2n^5 \leq  2n^5\left(\frac{\lambda_{d+1}}{\lambda_0 (1-2\epsilon)}\right)^{2K} \leq 2 n^5 \left(\frac{1-3\epsilon}{1-2\epsilon}\right)^{2K} \leq 2 n^5 \frac{\epsilon^6}{n^6}\leq \frac{\epsilon^2}{4}.\qedhere
\]
\end{proof}

The next two helper lemmas are to show that the projection of $\ww$ to $T_{\nu}- \tilde{T}$ is small. To do this, we introduce some more notations and one proposition. 
For any level $\nu$, we will use $q_{\nu}$ to denote,
\begin{equation}
    q_{\nu} \defeq \dim{T_{\nu}-\tilde{T}}.
\end{equation}
Let 
\begin{equation}\label{def:z}
\zz=\sum_{i}z_{i}\uu_{i}\text{ where }z_{i}=\lambda_{i}^{K}\alpha_{i}.
\end{equation} 
Therefore, $\ww=\zz/\|\zz\|$.

We now bound the norm of the projection of $\zz$ to $T_\nu - \Ttil$. The proof is based on the fact that $\zz$ is a {\it scaled} gaussian random vector, and the projection of a gaussian on a $q_{\nu}$-dimensional subspace should have norm proportional to $q_{\nu}$. 
\begin{proposition}\label{lem:projZ}
For $\zz$ as defined in~\eqref{def:z}, we have
\[
\|\Pi_{\nu}\zz\|^2 \leq \lambda_{a_{\nu}}^{2K} \cdot \sum_{j=a_{\nu}}^{b_{\nu}}\alpha_j^2.
\] 
Furthermore, if $q_{\nu}\geq 10 \log n$, then with probability at least $1-1/n^2$,
\[
\|\Pi_{\nu}\zz\|^2 \leq 2q_{\nu}\cdot \lambda_{a_{\nu}}^{2K}.
\]
\end{proposition}
\begin{proof}
Let $\Pi_{\nu}^{\text{full}}$ be a projection matrix to the subspace $T_{\nu}$. Recall from \Cref{def:Basis} that $\uu_{a_{\nu}},\dots,\uu_{b_{\nu}}$ form an orthonormal basis of $T_\nu$ and so we have $\Pi_{\nu}^{\text{full}}=\sum_{j=a_{\nu}}^{b_{\nu}}\uu_{j}\uu_{j}^{\top}$ and so $$\Pi_{\nu}^{\text{full}}\zz=\sum_{j=a_{\nu}}^{b_{\nu}}(\lambda_{j}^{K}\alpha_{j})\uu_{i}.$$ Since $T_{\nu}-\Ttil\subseteq T_{\nu}$, we have 
\[
\|\Pi_{\nu}\zz\|^{2}\le\|\Pi_{\nu}^{\text{full}}\zz\|^{2} = \lambda_{a_{\nu}}^{2K}\sum_{j=a_{\nu}}^{b_{\nu}}\alpha_{j}^{2},
\]
which proves the first part of the lemma. 

Before proving the second part, we consider the vector $\yy=\sum_{i}\alpha_{i}\uu_{i}$. Since $\alpha_{i}\sim N(0,1)$ for all $i$, we also have $\Pi_{\nu}\yy\sim N(0,1)$ is a gaussian in a $q_{\nu}$-dimensional space. So by \Cref{lem:NormG}, we have
\[
\|\Pi_{\nu}\yy\|^{2}\le2q_{\nu}
\]
with probability at least $1-e^{-q_{\nu}/4}\ge1-e^{-2\log n}=1-1/n^{2}$. 

To prove the second part, observe that $(\lambda_{a_{\nu}}^{K}\cdot\yy)$ ``dominates'' $\Pi_{\nu}^{\text{full}}\zz$ in every coordinate, i.e,. the coefficient of each $\uu_{i}$ in $(\lambda_{a_{\nu}}^{K}\cdot\yy)$ is at least that of $\Pi_{\nu}^{\text{full}}\zz$ for every $i$. Therefore, $\|\PP(\Pi_{\nu}^{\text{full}}\zz)\|\le\|\PP(\lambda_{a_{\nu}}^{K}\cdot\yy)\|$ for any projection matrix $\PP$. Since $T_{\nu}-\Ttil\subseteq T_{\nu}$, we have $\Pi_{\nu}\zz=\Pi_{\nu}\Pi_{\nu}^{\text{full}}\zz$. Therefore, we can conclude that 
\[
\|\Pi_{\nu}\zz\|^{2}=\|\Pi_{\nu}\Pi_{\nu}^{\text{full}}\zz\|^{2}\le\|\Pi_{\nu}(\lambda_{a_{\nu}}^{K}\cdot\yy)\|^{2}
\]
which is at most $\lambda_{a_{\nu}}^{2K}\cdot2q_{\nu}$ with probability at least $1-1/n^{2}$. \qedhere
\end{proof}
The following lemma shows that the projection of $\ww$ on $T_{\nu}-\Ttil$ is small when $\dim{T_\nu - \Ttil}:= q_\nu$ is roughly at most an $\epsilon$-factor of $\dim{T_\nu}:= d_\nu$, and $q_\nu$ is still at least logarithmic.
We will use this lemma to characterize the projection of $\ww$ on $T_{\nu}$ for $\nu \in \mathcal{I}$. 
Recall that $\Pi_\nu$ is a projection matrix that projects any vector to the space $T_{\nu}-\Ttil$.
\begin{lemma}\label{lem:GaussianProjD}

If $ 10 \log n \leq q_{\nu} \leq \frac{\epsilon}{300\log\frac{n}{\epsilon}}d_{\nu}$, then
\[
\Pr\left[\norm{\Pi_{\nu}\ww}^2 \leq \frac{\epsilon}{\log\frac{n}{\epsilon}} \right] \geq 1-\frac{2}{n^{2}}.
\]
\end{lemma}
\begin{proof}
Since $\ww=\zz/\|\zz\|$, it is equivalent to show that, with probability $1-\frac{2}{n^{2}}$, $$\|\Pi_{\nu}\zz\|^{2}\leq\frac{\epsilon}{\log\frac{n}{\epsilon}}\|\zz\|^{2}.$$ 

We first bound $\|\Pi_\nu \zz\|$ in terms of $d_\nu$. As $q_{\nu} \ge 10\log n$, by \Cref{lem:projZ}, we have with probability at least $1-1/n^{2}$, 
\[
\|\Pi_{\nu}\zz\|^{2}\leq\lambda_{a_{\nu}}^{2K}\cdot2q_{\nu}\le\lambda_{a_{\nu}}^{2K}\cdot\frac{\epsilon}{150\log\frac{n}{\epsilon}}d_{\nu}.
\]
where the second inequality follows from the assumption $q_{\nu}\leq\frac{\epsilon}{300\log\frac{n}{\epsilon}}d_{\nu}$. 

Next, we bound $d_{\nu}$ in terms of $\|z\|$. Consider the $d_{\nu}$-dimensional gaussian vector with coordinate $\alpha_{i}$ for $i=a_{\nu},\dots,b_{\nu}$. Applying \Cref{lem:NormG} to this vector with $\delta=1/10$, we have $\Pr[\sum_{i=a_{\nu}}^{b_{\nu}}\alpha_{i}^{2}\geq (1-\frac{1}{2})d_{\nu}]\geq1-e^{-d_{\nu}/100}\geq1-\frac{1}{n^{2}}$ where the last inequality used that $d_{\nu}\ge 3000\log n$. With probability $1-1/n^{2}$, we now have
\[
d_{\nu}\le2\sum_{i=a_{\nu}}^{b_{\nu}}\alpha_{i}^{2}=2\sum_{i=a_{\nu}}^{b_{\nu}}\frac{z_{i}}{\lambda_{i}^{2K}}\le\frac{2}{\lambda_{b_{\nu}}^{2K}}\|\zz\|^{2}
\]
Combining the two inequalities, we can conclude that, with probability at least $1-2/n^{2}$, 
\[
\|\Pi_{\nu}\zz\|^{2}\le\left(\frac{\lambda_{a_{\nu}}}{\lambda_{b_{\nu}}}\right)^{2K}\frac{\epsilon}{75\log\frac{n}{\epsilon}}\|\zz\|^{2}\le\frac{\epsilon}{\log\frac{n}{\epsilon}}\|\zz\|^{2}
\]
as desired. To see the last inequality, recall from \Cref{def:Basis} that $\lambda_{a_{\nu}}\leq\left(1-\frac{\nu\epsilon}{5\log\frac{n}{\epsilon}}\right)\lambda_{0}$ and $\lambda_{b_{\nu}}\geq\left(1-\frac{(\nu+1)\epsilon}{5\log\frac{n}{\epsilon}}\right)\lambda_{0}$. So $\frac{\lambda_{a_{\nu}}}{\lambda_{b_{\nu}}}\le1+\frac{\epsilon}{2\log\frac{n}{\epsilon}}$ and, hence, 
\[
\left(\frac{\lambda_{a_{\nu}}}{\lambda_{b_{\nu}}}\right)^{2K}\leq\left(1+\frac{\epsilon}{2\log\frac{n}{\epsilon}}\right)^{2K}\le e^{4}\approx54.6. \qedhere
\]
\end{proof}

We next prove that for all $\nu\notin \mathcal{I}$, arbitrary projections of $\ww$ on $T_{\nu}-\Ttil$ are always small. This proof uses a similar idea as that of the previous lemma and the main difference is that we can use the fact that $d_{\nu}$ is small for $\nu\notin \mathcal{I}$ to additionally show that the projection of $\ww$ is small even for small dimensional arbitrary subspaces of $T_{\nu}$. Recall that $\Pi_\nu$ is a projection matrix to the space $T_{\nu}-\Ttil$.
\begin{restatable}{lemma}{NotImp}\label{lem:NotImp}
For any non-important level, $\nu \notin \cal{I}$,
\[
\Pr\left[\norm{\Pi_{\nu}\ww}^2 \leq \frac{\epsilon}{\log\frac{n}{\epsilon}} \right] \geq 1-\frac{1}{n^{2}}.
\]
\end{restatable}

\begin{proof}
Again, it is sufficient to prove for $\zz$ as defined in \Cref{def:z},
\[
\Pr\left[\norm{\Pi_{\nu}\zz}^2 \leq \frac{\epsilon}{\log\frac{n}{\epsilon}}\|\zz\|^2 \right] \geq 1-\frac{1}{n^{2}}.
\]
In this proof, we consider the case of $q_{\nu}\geq 20\log n$ and $q_{\nu}<20\log n$ separately. Let us first look at $q_{\nu}\geq 20\log n$.
\paragraph{Case $q_{\nu}\geq 20\log n$:}
Our strategy will be to first bound $\|\Pi_{\nu}\zz\|^2$ by $d_{\nu}$, which can be further bounded by $\sum_{\nu'<\nu}d_{\nu'}$.
From \Cref{lem:projZ}, with probability at least $1-1/n^2$,
\[
\|\Pi_{\nu}\zz\|^2 \leq 2q_{\nu}\lambda_{a_{\nu}}^{2K}.
\]
Since $T_{\nu}-\tilde{T}\subseteq T_{\nu}$, $q_{\nu}\leq d_{\nu}$. Furthermore, since $\nu \notin \mathcal{I}$, $d_{\nu}\leq \frac{\epsilon}{600\log^3 \frac{n}{\epsilon}}\sum_{\nu'<\nu}d_{\nu'}$. Using these bounds, we then have with probability $1-1/n^2$,  
\[
\|\Pi_{\nu}\zz\|^2 \leq \frac{\epsilon}{300\log^3 \frac{n}{\epsilon}}\lambda_{a_{\nu}}^{2K}\sum_{\nu'<\nu}d_{\nu'}.
\]
We next bound $\sum_{\nu'<\nu}d_{\nu'}$ in terms of $\|\zz\|$. Consider the $\sum_{\nu'<\nu}d_{\nu'}$ dimensional gaussian vector with coordinates $\alpha_i$, for $i = 1, \cdots, b_{\nu-1}$. Applying \Cref{lem:NormG} to this vector with $\delta=1/3$ gives,
\[
\Pr\left[\sum_{i=1}^{b_{\nu-1}} \alpha_i^2 \geq \left(1-\frac{8}{9}\right)\sum_{\nu'<\nu}d_{\nu'}\right] \geq 1- e^{-\frac{\sum_{\nu'<\nu}d_{\nu'}}{9}} \geq 1-\frac{1}{n^2}.
\]
In the last inequality we used that $\sum_{\nu'<\nu}d_{\nu'}\geq d_{\nu}$ and $d_{\nu}\geq q_{\nu}\geq 20\log n$. We now know that with probability at least $1-1/n^2$,
\[
\sum_{\nu'<\nu}d_{\nu'} \leq 9\sum_{i=1}^{b_{\nu-1}}\alpha_i^2 = 9\sum_{i=1}^{b_{\nu-1}}\frac{\zz_i^2}{\lambda_i^{2K} }\leq 9\frac{\|\zz\|^2}{\lambda_{b_{\nu-1}}^{2K}}.
\]
Therefore, with a probability of at least $1-2/n^2$,
\[
\|\Pi_{\nu}\zz\|^2 \leq \frac{\epsilon}{15\log^3 \frac{n}{\epsilon}}\left(\frac{\lambda_{a_{\nu}}}{\lambda_{b_{\nu-1}}}\right)^{2K}\|\zz\|^2 \leq  \frac{\epsilon}{15\log^3 \frac{n}{\epsilon}}\|\zz\|^2.
\]
The last inequality follows from the fact $\lambda_{a_{\nu}} \leq \lambda_{b_{\nu-1}}$.
\paragraph{Case $q_{\nu} <20 \log n$:} In this case, since $q_{\nu}<20\log n$, we apply the first part of \Cref{lem:projZ} to get
\[
\|\Pi_{\nu}\zz\|^2\leq  \lambda_{a_{\nu}}^{2K} \sum_{j=a_{\nu}}^{b_{\nu}}\alpha_j^2.
\]
Observe that in this case $d_{\nu}$ can be less than $20\log n$. We also know that since $\nu \notin \mathcal{I}$, we can bound $d_{\nu}$ by $\sum_{\nu'<\nu}d_{\nu'}$. In our analysis, we consider the value of $\sum_{\nu'<\nu}d_{\nu}$ and further split it into two parts based on whether $\sum_{\nu'<\nu}d_{\nu}$ is large or small.
\begin{itemize}
    \item $\sum_{\nu'<\nu} d_{\nu'} \geq 20 \log n$: Our strategy would be to first bound $\|\Pi_{\nu}\zz\|^2$ by $\sum_{\nu'<\nu}d_{\nu}$, and then bound $\sum_{\nu'<\nu}d_{\nu}$ by $\|\zz\|^2$. Note that since $\alpha_i$'s are gaussian random variables, $\sum_{j=a_{\nu}}^{b_{\nu}}\alpha_j^2$ follows a $\chi^2_k$ distribution with $k = d_{\nu}$. Therefore, $\sum_{j=a_{\nu}}^{b_{\nu}}\alpha_j^2 \leq d_{\nu} \log n$ with probability at least $1-1/n^2$. Further since $d_{\nu}\leq \frac{\epsilon}{600 \log^3 \frac{n}{\epsilon}}\sum_{\nu'<\nu} d_{\nu'}$, we get with probability at least $1-2/n^2$,
    \[
    \|\Pi_{\nu}\zz\|^2\leq  \lambda_{a_{\nu}}^{2K} \cdot \log n\cdot d_{\nu} \leq \lambda_{a_{\nu}}^{2K}\frac{\epsilon}{600\log \frac{n}{\epsilon}}\sum_{\nu'<\nu} d_{\nu'}.
    \]
    Now, since $\sum_{\nu'<\nu} d_{\nu'} \geq 20 \log n$ we use Lemma~\ref{lem:NormG} again with $\delta = 1/3$, on a vector with coordinates $\alpha_i$'s for $i = 1,\cdots,b_{\nu-1}$ to get with probability at least $1-1/n^2$, 
    \[
    \sum_{\nu'<\nu} d_{\nu'} \leq 9\sum_{i=1}^{b_{\nu-1}}\alpha_i^2=9\sum_{i=1}^{b_{\nu-1}}\frac{z_i^2}{\lambda_i^{2K}} \leq 9\frac{\|\zz\|^2}{\lambda_{b_{\nu-1}}^{2K}}.
    \]
    Plugging this back, we get with probability at least $1-3/n^2$,
    \[
    \|\Pi_{\nu}\zz\|^2\leq \frac{\epsilon}{60\log \frac{n}{\epsilon}} \left(\frac{\lambda_{a_{\nu}}}{\lambda_{b_{\nu-1}}}\right)^{2K}\|\zz\|^2 \leq\frac{\epsilon}{60\log \frac{n}{\epsilon}}\|\zz\|^2.
\]
Last inequality follows from $\lambda_{a_{\nu}} \leq \lambda_{b_{\nu-1}}$.
\item $\sum_{\nu'<\nu} d_{\nu'} < 20 \log n$: Since $\nu\notin \mathcal{I}$, we know that 
\[
d_{\nu}\leq \frac{\epsilon}{600 \log^3\frac{n}{\epsilon}}\sum_{\nu'<\nu}d_{\nu'}.
\]
Since $\sum_{\nu'<\nu} d_{\nu'} < 20 \log n$, we must then have that,
\[
d_{\nu}\leq \frac{\epsilon}{30 \log^2\frac{n}{\epsilon}} <1.
\]
Since the dimension $d_{\nu}$ must be an integer, it must be the case that $d_{\nu} = 0$, and therefore, $\|\Pi_{\nu}\zz\| = 0.$
\end{itemize}
\end{proof}

\paragraph{Proof of \Cref{cl:progress}.}
 We are now ready to finally prove \Cref{cl:progress}.

\begin{proof}

We want to show that if $\lambda_{\max}(\AA)\geq 1- \epsilon$ and $\dim{T_{\nu}-\tilde{T}}$ is small for all $\nu\in \mathcal{I}$ as stated in \Cref{cl:progress}, then $\ww^{\top}\VV\ww\leq 35 \epsilon$ for all $\ww\in \WW$ with high probability. Recall that
\[
\ww^{\top}\VV\ww = \ww^{\top}\VV_{\tilde{T}}\ww+\sum_{\nu = 0}^{15\log\frac{n}{\epsilon}-1}\ww^{\top}\VV_{T_{\nu} -\tilde{T}}\ww+ \ww^{\top}\VV_{\overline{T}}\ww.
\]
Let us upper bound each term in the sum below.
\begin{enumerate}
    
    \item \textbf{$\ww^\top\VV_{\tilde{T}}\ww$:} 
    We have 
    \[
        \ww^\top \VV_{\Ttil} \ww = (\Pi_{\tilde{T}}\ww)^{\top}\VV\Pi_{\tilde{T}}\ww =
        (\Pi_{\tilde{T}}\ww)^{\top}\AA\Pi_{\tilde{T}}\ww - (\Pi_{\tilde{T}}\ww)^{\top}\AAtil\Pi_{\tilde{T}}\ww. 
    \]
    From the definition of $\tilde{T}$ (see \Cref{def:SpaceA}), we know that $(\Pi_{\tilde{T}}\ww)^{\top}\AAtil\Pi_{\tilde{T}}\ww \geq \left(1-10\epsilon\right)\lambda_0\|\Pi_{\tilde{T}}\ww\|^2 $ because $\Pi_{\tilde{T}}\ww \in \Ttil = \Span(3\eps,\AAtil).$
    We also know that $(\Pi_{\tilde{T}}\ww)^{\top}\AA\Pi_{\tilde{T}}\ww \leq \lambda_0 \|\Pi_{\tilde{T}}\ww\|^2$. So 
    \begin{equation}\label{eq:V2}
    \ww^{\top}\VV_{\tilde{T}}\ww \leq 10\epsilon\lambda_0 \|\Pi_{\tilde{T}}\ww\|^2\leq 10\epsilon(1+\epsilon)\|\ww\|^2 =10\epsilon (1+\epsilon).
  \end{equation}
\item \textbf{$\ww^\top\VV_{\overline{T}}\ww$:} Since $\ww\in \WW$, from Lemma~\ref{lem:boundLowEV}, with probability at least $1-\frac{3}{n^2}$, $\|\Pi_{\overline{T}}\ww\|\leq \epsilon/2$. Now, 
    \begin{equation}\label{eq:V1}
        \ww^{\top}\VV_{\overline{T}}\ww = (\Pi_{\overline{T}}\ww)^{\top}\VV(\Pi_{\overline{T}}\ww)\leq \|\Pi_{\overline{T}}\ww\|^2\|\VV\| \leq \frac{\epsilon^2}{4}\lambda_0 \leq \frac{\epsilon^2}{2},
    \end{equation}
    where we used that $\|\VV\|\leq \lambda_0$ since $\AAtil = \AA-\VV \succeq 0$ and $\lambda_0\leq 1+\epsilon/\log n$.
    \item \textbf{$\ww^\top\VV_{T_{\nu} -\tilde{T}}\ww$ when $\nu \in \mathcal{I}$:} Note that the dimension of the space $T_{\nu} -\tilde{T}$ is small. We now have,
    \[
    \ww^{\top}\VV_{T_{\nu} -\tilde{T}}\ww = \ww^{\top}\Pi_{\nu}\VV\Pi_{\nu}\ww.
    \]
    We will now consider the large $d_{\nu}$ and small $d_{\nu}$ cases separately.
    \paragraph{Large dimension: $d_{\nu}\geq \frac{3000\log n\log\frac{n}{\epsilon}}{\epsilon}$.}
    In this case, $\dim{T_{\nu} -\tilde{T}} \leq \frac{\epsilon}{300\log\frac{n}{\epsilon}} d_{\nu}$. We can now apply Lemma~\ref{lem:GaussianProjD}, which gives with probability at least $1-\frac{2}{n^2}$,
    \[
\|\Pi_{\nu}\ww\|^2 \leq \frac{\epsilon}{\log\frac{n}{\epsilon}}.
    \]

    Now, using this value,
    \begin{equation}\label{eq:V3}
    \ww^{\top}\VV_{T_{\nu} -\tilde{T}}\ww \leq \|\Pi_{\nu}\ww\|^2 \|\VV\| \leq \frac{\epsilon}{\log\frac{n}{\epsilon}} \|\VV\|\leq  \frac{\epsilon}{\log\frac{n}{\epsilon}} \lambda_0 \leq  \frac{\epsilon(1+\epsilon)}{\log\frac{n}{\epsilon}}.
    \end{equation}
    As in case 2, we again used the fact that $\|\VV\| \le \lambda_0 \le (1+\eps)$.
    \paragraph{Small dimension: $d_{\nu}<\frac{3000\log n\log\frac{n}{\epsilon}}{\epsilon}$.}
    In this case, $\dim{T_{\nu} -\tilde{T}}<1.$ Therefore, the space $T_{\nu} -\tilde{T}$ is empty and as a result, 
    \begin{equation}\label{eq:V4}
    \ww^{\top}\VV_{T_{\nu} -\tilde{T}}\ww  = 0.
    \end{equation}
\item \textbf{$\ww^\top\VV_{T_{\nu} -\tilde{T}}\ww$ when $\nu \notin \mathcal{I}$:}
From Lemma~\ref{lem:NotImp}, $\|\Pi_{\nu}\ww\|^2 \leq \frac{\epsilon}{\log\frac{n}{\epsilon}}$ with probability at least $1-1/n^2$. Since $\|\VV\|\leq \lambda_0\leq 1+\epsilon$, we get,
\begin{equation}\label{eq:V5}
\ww^{\top}\VV_{T_{\nu} -\tilde{T}}\ww \leq \|\Pi_{\nu}\ww\|^2 \|\VV\| \le \frac{\epsilon(1+\epsilon)}{\log\frac{n}{\epsilon}}.
\end{equation}
\end{enumerate}
We now combine all the cases. We have for both large and small $d_{\nu}$ from Equations~\eqref{eq:V2},\eqref{eq:V1},\eqref{eq:V3},\eqref{eq:V4} and \eqref{eq:V5}, with probability at least $1-\frac{40\log\frac{n}{\epsilon}}{n^2}$ for any $\ww\in \WW$,
\[
\ww^{\top}\VV\ww \leq \ww^{\top}\VV_{\overline{T}}\ww + \ww^{\top}\VV_{\tilde{T}}\ww + \sum_{\nu = 0}^{15\log\frac{n}{\epsilon}-1}\ww^{\top}\VV_{T_{\nu} -\tilde{T}}\ww \leq \frac{\epsilon^2}{2} + 10\epsilon(1+\epsilon) + 10\epsilon(1+\epsilon) \leq 35\epsilon.
\]

\end{proof}

\section{Conditional Lower Bounds for an Adaptive Adversary}\label{sec:Adap}

In this section, we will prove a conditional hardness result for algorithms against adaptive adversaries. In particular, we will prove \Cref{thm:lower}.
Consider \Cref{alg:red} for solving \Cref{prob:factor}. 
The only step in \Cref{alg:red} whose implementation is not specified is Line~\ref{line:approx eig}. We will implement this step using an algorithm for \Cref{prob:dyn}.

\begin{algorithm}
\caption{Algorithm for Checking PSDness}\label{alg:red}
 \begin{algorithmic}[1]
\Procedure{CheckPSD}{$\delta,\kappa,\AA$}
\State $\epsilon \leftarrow \min\{1- n^{-o(1)}, (1-\delta)/(1+\delta)\}$
\State $T \leftarrow \frac{2n}{\epsilon(1-\epsilon)^2}\log \frac{\kappa}{\delta}$
\State $\AA_0 \leftarrow \AA$
\State $\mu_0 = 0, \ww_0 = 0$
\For{$t  = 1,2,\cdots,T$}
\State $\AA_t = \AA_{t-1} - \frac{\mu_{t-1}}{10}\ww_{t-1}\ww_{t-1}^{\top}$\label{line:red update}
\State $(\mu_t,\ww_t)\leftarrow \epsilon$-approximate maximum eigenvalue and eigenvector of $\AA_t$ (Equations~\eqref{eq:epsEigvalue},\eqref{eq:epsEigvec})\label{line:approx eig}
\If{$\mu_t<0$}
\State \Return {\sc False}:$\AA$ is not PSD
\EndIf
\EndFor
\State {$\sigma^2 \gets $PowerMethod($\epsilon,\AA_T^{\top}\AA_T$)}\label{line:last check}
\If{$0\leq \sigma \leq \frac{(1+\epsilon)\mu_1\delta}{\kappa}$}\label{line:red lastcheck}
\State \Return $\XX=\frac{1}{\sqrt{10}}\begin{bmatrix}\sqrt{\mu_1}\ww_1 & \sqrt{\mu_2}\ww_2 & \cdots & \sqrt{\mu_T}\ww_T\end{bmatrix}$
\Else
\State \Return {\sc False}:$\AA$ is not PSD
\EndIf
\EndProcedure 
 \end{algorithmic}
\end{algorithm}

\paragraph{High-level idea.}
 Overall for our hardness result, we use the idea that an adaptive adversary can use the maximum eigenvectors returned to perform an update. This can happen $n$ times and in the process, we would recover the entire eigen-decomposition of the matrix, which is hard. Now consider Algorithm~\ref{alg:red}. We claim that \Cref{alg:red} solves \Cref{prob:factor}. 
At the first glance, this claim looks suspicious because the input matrix for \Cref{prob:factor} might not be PSD, but the dynamic algorithm for \Cref{prob:dyn} at Line~\ref{line:approx eig} has any guarantees only when the matrices remain PSD. 
However, the reduction does work by crucially exploiting \Cref{prop:assume}. The high-level idea is as follows. 
\begin{itemize}
    \item If the input matrix $\AA$ is initially PSD, then we can show that $\AA_t$ remains PSD for all $t$  by exploiting \Cref{prop:assume}, (see \Cref{lem:orthoUpdate}). So, the approximation guarantee of the algorithm at Line~\ref{line:approx eig} is valid at all steps.
    From this guarantee, $\|\AA_T\|$ must be tiny since we keep decreasing the approximately maximum eigenvalues (see \Cref{lem:RedAns}). At the end, the reduction will return $\XX$.
    \item If the input matrix $\AA$ is initially \emph{not} PSD, there must exist a direction $\vv$ such that $\vv^{\top}\AA\vv <0$. Since in the reduction, we update $\AA_T = \AA - \WW$ for some $\WW \succeq 0$, we must have that $\vv^{\top}\AA_T\vv <\vv^{\top}\AA\vv$. That is, this negative direction remains negative or gets even more negative. It does not matter at all what guarantees the algorithm at Line~\ref{line:approx eig} has. We still have that $\|\AA_T\|$ cannot be tiny. We can distinguish whether $\|\AA_T\|$ is tiny or not using the static power method at Line~\ref{line:last check}, and, hence, we will return {\sc False} in this case (see \Cref{lem:RedAns}).
\end{itemize}


\paragraph{Analysis.}
We prove the guarantees of the output of Algorithm~\ref{alg:red} when $\ww_t$'s satisfy \Cref{def:super} for all $t$.

\begin{lemma}\label{lem:orthoUpdate}
In Algorithm~\ref{alg:red}, let $\ww_t$'s, $t=1,\cdots, T$ be generated such that they additionally satisfy \Cref{def:super}. If $\AA_0\succeq 0$, then $\AA_t \succeq 0$ for all $t$.
\end{lemma}
 We would like to point out that our parameter $\epsilon$ is quite large. This just implies that our reduction can work even if we find crude approximations to the maximum eigenvector as long as this is along the directions with large eigenvalue, since $\ww$ also has to satisfy \Cref{def:super}.
\begin{proof}
We will prove this by induction. This is true in the beginning because $\AA_0 \succeq 0$ by assumption. Let us assume $\AA_{t}\succeq 0$. We now look at $\AA_{t+1} = \AA_{t}- \frac{\mu_{t}}{10} \ww_{t}\ww_{t}^{\top}$.
 For simplicity, we use $\lambda_i,\uu_i$ to denote the  $i^{th}$ eigenvalue and eigenvector of $\AA_{t}$, and $\mu = \mu_t$, $\ww = \ww_t$. By \Cref{eq:epsEigvalue}, $\lambda_1 \geq \mu \geq (1-\epsilon)\lambda_1$. Now, for any $\yy$
\begin{equation}\label{eq:psdT}
\yy^{\top}\AA_{t+1}\yy = \yy^{\top}\AA_{t}\yy - \frac{\mu}{10} (\ww^{\top}\yy)^2 = \sum_{i=1}^n \lambda_i (\uu_i^{\top}\yy)^2 -  \frac{\mu}{10}\left(\sum_{i=1}^n (\ww^{\top}\uu_i)(\yy^{\top}\uu_i)\right)^2 .
\end{equation}
We will upper bound the second term. We know that $\ww$ satisfies \Cref{def:super}. Let $d$ be such that $\lambda_i \leq \lambda_1/2$ for all $i \geq d$.
\begin{align*}
\left(\sum_{i=1}^n (\ww^{\top}\uu_i)(\yy^{\top}\uu_i)\right)^2 & \leq 2 \left(\sum_{i=1}^d(\ww^{\top}\uu_i)(\yy^{\top}\uu_i)\right)^2 + 2 \left(\sum_{i=d}^n(\ww^{\top}\uu_i)(\yy^{\top}\uu_i)\right)^2\\
& \substack{(a)\\ \leq} 2 \sum_{i=1}^d (\ww^{\top}\uu_i)^2\sum_{i=1}^d (\yy^{\top}\uu_i)^2 + 2(n-d) \sum_{i=d}^n  (\ww^{\top}\uu_i)^2(\yy^{\top}\uu_i)^2\\
& \substack{(b)\\ \leq} 2 \sum_{i=1}^d (\yy^{\top}\uu_i)^2 + 2(n-d)\sum_{i=d}^n \frac{\lambda_i(\yy^{\top}\uu_i)^2}{\lambda_1 n^2}\\
& \substack{(c)\\ \leq}  \frac{4}{\lambda_1}\sum_{i=1}^d \lambda_i(\yy^{\top}\uu_i)^2 + \sum_{i=d}^n \frac{\lambda_i(\yy^{\top}\uu_i)^2}{\lambda_1 n}.
\end{align*}
In the above, $(a)$ follows by applying Cauchy Schwarz to both terms, $(b)$ follows by using $\|\ww\| = 1$ in the first term and applying Property~\ref{def:super} on the second term, and $(c)$ follows by using the fact that for all $i \leq d$, $\lambda_i \geq \lambda_1/2$ in the first term. 

Using this in \eqref{eq:psdT},
\begin{align*}
\yy^{\top}\AA_{t+1}\yy & \geq \sum_{i=1}^n \lambda_i (\uu_i^{\top}\yy)^2 - \left(\frac{4\mu}{10\lambda_1}\sum_{i=1}^d \lambda_i (\uu_i^{\top}\yy)^2 + \frac{\mu}{10}\sum_{i=d}^n \frac{\lambda_i(\yy^{\top}\uu_i)^2}{\lambda_1 n}\right)\\
& = \sum_{i=1}^d \left(1 - \frac{2\mu}{5\lambda_1}\right)\lambda_i (\uu_i^{\top}\yy)^2  + \sum_{i=d}^n \left(1 - \frac{\mu}{10n\lambda_1}\right)\lambda_i (\uu_i^{\top}\yy)^2  \\
& \geq 0.
\end{align*}
We have thus shown that $\AA_{t+1}$ is psd, as required.
\end{proof}

\begin{lemma}\label{lem:RedAns}
In Algorithm~\ref{alg:red}, let $\ww_t$'s, $t=1,\cdots, T$ be generated such that they additionally satisfy \Cref{def:super}. 
\begin{itemize}
    \item If $\AA\succeq 0$, then \Cref{alg:red} returns $\XX$ such that $\|\AA-\XX\XX^{\top}\|\leq \delta \min_{\|\xx\|=1}\|\AA\xx\|$. 
    \item If $\AA$ is not psd, then \Cref{alg:red} returns {\sc False}.
\end{itemize}
\end{lemma}
\begin{proof}
We prove the first part first. Let $\AA_0\succeq 0$ with maximum eigenvalue $\lambda_1$, and $t$ be the iterate such that the maximum eigenvalue of $\AA_{t-1}$ is at least $(1-\epsilon)\lambda_1$ and the maximum eigenvalue of $\AA_t$ is at most $(1-\epsilon)\lambda_1$. We also know that $\AA_t\succeq 0$ for all $t$ from Lemma~\ref{lem:orthoUpdate}.

In this case, $\mu_i \geq (1-\epsilon)^2\lambda_1$ for all $i\leq t-1$ (since $\mu_i$ is an $\epsilon$-max eigenvalue and the max eigenvalue is at least $(1-\epsilon)\lambda_1$) and we must have that $\Tr[\AA_t]\leq \Tr[\AA] - (1-\epsilon)^2\lambda_1(t-1)$. Also,
\[
\lambda_{\max}(\AA_t) \leq \Tr[\AA_t]\leq \Tr[\AA] - (1-\epsilon)^2\lambda_1(t-1) \leq  n\lambda_1 - (1-\epsilon)^2\lambda_1(t-1).
\]
Thus we require at most $t = 2n/(1-\epsilon)^2$ iterates after which $\AA_t$ will have a maximum eigenvalue at most $(1-\epsilon)\lambda_1$. 

Our algorithm runs for $T = \frac{2n}{\epsilon(1-\epsilon)^2}\log \frac{\kappa}{\delta}$ and every $2n/(1-\epsilon)^2$ iterations we decrease the maximum eigenvalue by a factor of at least $1-\epsilon$. Finally, the maximum eigenvalue of $\AA_T$ is therefore at most,
\[
\lambda_1 \cdot (1-\epsilon)^{\frac{T(1-\epsilon)^2}{2n}}=\lambda_1 \cdot (1-\epsilon)^{\frac{1}{\epsilon}\log \frac{\kappa}{\delta}} \leq \lambda_1 e^{-\log \frac{\kappa}{\delta}}\leq \lambda_1 \frac{\delta}{\kappa} = \delta \cdot \min_{\|\xx\|=1}\|\AA\xx\|, 
\] 
as required. We have proved that, if $\AA\succeq 0$, then the maximum eigenvalue of $\AA-\XX\XX^{\top}$ is at most $\lambda_1\delta/\kappa \leq \delta \min_{\|\xx\|=1}\|\AA\xx\|$.%

Next, if $\AA$ is not psd, then we will prove that $\sigma\ge\frac{(1-\epsilon)\mu_{1}}{\kappa}.$ This implies that $\sigma>\frac{(1+\epsilon)\mu_{1}\delta}{\kappa}$ from the value of $\epsilon$.  
and so the algorithm will return $\textsc{False}$ by the condition at Line~\ref{line:red lastcheck}. Since $\AA$ is not psd, there exists a unit vector $\vv$ such that $\vv^{\top}\AA\vv<0$. We can lower bound $\sigma$ as follows.
\[
\sigma\ge(1-\epsilon)\|\AA_{T}\|\ge(1-\epsilon)|\vv^{\top}\AA_{T}\vv|.
\]
On the other hand, we have
\[
|\vv^{\top}\AA_{T}\vv|\ge|\vv^{\top}\AA\vv|\ge\min_{\|\xx\|=1}\xx^{\top}\AA\xx\ge\frac{\mu_{1}}{\kappa}
\]
where the first inequality is because $\vv^{\top}\AA_{T}\vv=\vv^{\top}\AA\vv-\|\XX^{\top}\vv\|^{2}<\vv^{\top}\AA\vv<0$. Combining the two inequalities, we get $\sigma\ge\frac{(1-\epsilon)\mu_{1}}{\kappa}$, which concludes the proof.
 \end{proof}

\paragraph{Proof of Theorem~\ref{thm:lower}.} 
We are now ready to prove our conditional lower bound.
\begin{proof} 
Let $\calM(\epsilon,\AA_0,\vv_1,\cdots,\vv_T)$ denote an algorithm for \Cref{prob:dyn} that maintains an $\epsilon$-max eigenvalue~\eqref{eq:epsEigvalue}, $\mu_t$, and eigenvector~\eqref{eq:epsEigvec}, $\ww_t$, for matrices $\AA_t = \AA_{t-1}-\vv_t\vv_t^{\top}$ such that $\ww_t$'s satisfy \Cref{def:super}. We will show that if the total update time of $\calM$ is $n^{o(1)}\cdot \left(nnz(\AA_0) + \sum_{t=1}^T nnz(\vv_i)\right)$, then there is an $n^{2+o(1)}$-time algorithm for \Cref{prob:factor} which contradicts \Cref{conj:decomp is hard}.

Given an instance $(\delta,\kappa,\AA)$ of \Cref{prob:factor}, 
we will run \Cref{alg:red} where Line~\ref{line:approx eig} is implemented using $\calM$. We will generate the input and the update sequence for $\calM$ as follows.
Set $\AA_0 \gets \AA$. Set $\eps$ and $T$ according to \Cref{alg:red}. For $1\le t\le T$, we set $\vv_t = \frac{1}{\sqrt{10}}\cdot \sqrt{\mu_{t-1}} \ww_{t-1}$ according to Line~\ref{line:red update} of \Cref{alg:red}. We note that this is a valid update sequence for \Cref{prob:dyn} when $\AA\succeq 0$ since from \Cref{lem:orthoUpdate}, 
 if $\AA\succeq 0$ then, $\AA_t = \AA - \sum_{i=0}^{t-1}\frac{\mu_i}{10}\ww_i\ww_i^{\top}\succeq 0$. 

Now, we describe what we return as an answer for \Cref{prob:factor}.  From \Cref{lem:RedAns} if $\AA$ is not PSD, then Algorithm~\ref{alg:red} returns \textsc{False} and reports the matrix is not PSD. Additionally if $\AA\succeq 0$, the algorithm returns $\XX=\frac{1}{\sqrt{10}}\begin{bmatrix}\sqrt{\mu_1}\ww_1 & \sqrt{\mu_2}\ww_2 & \cdots & \sqrt{\mu_T}\ww_T\end{bmatrix}$ as a certificate that $\AA$ is PSD. This completes the reduction from \Cref{prob:factor} to \Cref{prob:dyn}.

The total time required by the reduction is 
\[
O\left(n^{o(1)}\cdot (nnz(\AA) + \sum_{t=1}^T nnz(\ww_t))\right) \leq O\left(n^{o(1)}\cdot (n^2 + T\cdot n)\right) \leq O(n^{2+o(1)}\log\frac{\kappa}{\delta}),
\]
which is at most $n^{2+o(1)}$ when $\frac{\kappa}{\delta} \leq \poly(n)$. 

To conclude, we indeed obtain an $n^{2+o(1)}$-time algorithm for \Cref{prob:factor}. Assuming \Cref{conj:decomp is hard}, the algorithm $\calM$ cannot have $n^{o(1)}\cdot \left(nnz(\AA_0) + \sum_{t=1}^T nnz(\vv_i)\right)$ total update time.
\end{proof}

\section{Conclusion and Open Problems}\label{sec:open}

\paragraph{Upper Bounds.}
We have presented a novel extension of the power method to the dynamic setting. Our algorithm from \Cref{thm:upper} maintains a multiplicative $(1+\epsilon)$-approximate maximum eigenvalue and eigenvector of a positive semi-definite matrix that undergoes decremental updates from an oblivious adversary. The algorithm has polylogarithmic amortized update time per non-zeros in the updates. 

Our algorithm is simple, but our analysis is quite involved. While we believe a tighter analysis that improves our logarithmic factors is possible, it is an interesting open problem to give a simpler analysis for our algorithm. 
Other natural questions are whether we can get similar algorithms in incremental or fully dynamic settings and whether one can get a worst-case update time.

\paragraph{Lower Bounds.}
We have shown a conditional lower bound for a class of algorithms against an adaptive adversary in \Cref{thm:lower}. It would also be very exciting to generalize our lower bound to hold for any algorithm against an adaptive adversary, as that would imply an oblivious-vs-adaptive separation for a natural dynamic problem. 

\paragraph{Incremental Updates.} We believe that the corresponding incremental updates problem, i.e., we update the matrix as $\AA_t \gets \AA_{t-1}+\vv_t\vv_t^{\top}$ cannot be solved in polylogarithmic amortized update time, even when the update sequence $\vv_t$'s are from an oblivious adversary. 
At a high level, the incremental version of our problem seems significantly harder for the following reasons. When we perform decremental updates to a matrix, the new maximum eigenvector must be a part of the eigenspace spanned by the original maximum eigenvectors. Furthermore, it is easy to detect whether the maximum eigenvalue has gone down as we have shown in our paper. For the incremental setting, it is possible that after an update the maximum eigenvalue has gone up and the new maximum eigenvector is a direction that was not the previous one or the update direction and in such cases we cannot really detect this quickly with known information on previous eigenvectors and update directions. This can also happen $n$ times and in every such case, we have to compute the eigenvalue and eigenvector from scratch. Therefore, we leave lower bounds and algorithms for incremental setting as an open problem. 

\paragraph{Dynamic SDPs.}
As discussed in \Cref{sec:dyn psd}, \Cref{thm:upper} can be viewed as a starting point towards a dynamic algorithm for general positive semi-definite programs. Can we make further progress? The dynamic semi-definite program problem, even with just two matrix constraints, already seems to be challenging. 

One promising approach to attack this problem is to dynamize the matrix multiplicative weight update (MMWU) method for solving a packing/covering SDP \cite{peng2012faster} since the corresponding approach was successful for linear programs -- the near-optimal algorithms of \cite{bhattacharya2023dynamic} are essentially dynamized multiplicative weight update methods for positive linear programs. However, in our preliminary study exploring this approach, we could only obtain an algorithm that solves Problem~\ref{prob:dynsdp}, which has a single matrix constraint, and  solves~\Cref{prob:dyn} partially, i.e., maintains an approximate eigenvalue only. The main barrier in this approach is that the algorithm requires maintaining the exponential of the sum of the constraint matrices, and to do this fast, we require that for any two constraint matrices $\AA$ and $\BB$, $e^{\AA+\BB}= e^{\AA} e^{\BB}$ which only holds when $\AA$ and $\BB$ commute i.e., $\AA\BB = \BB\AA$. Note that when $\AA$ and $\BB$ are diagonal, this is true; therefore, we can obtain the required algorithms for positive LPs. Even when we have just two constraint matrices where one of them is a diagonal matrix, this remains an issue as the matrices still do not commute. 

	 \newpage
	\printbibliography

	 \newpage
	\appendix
	\section{Connections to Dynamic Positive Semi-definite Programs}
\label{sec:dyn psd}
This section discusses connections between our \Cref{prob:dyn} and the dynamic versions of positive semi-definite programs. Using this connection, we conclude that \Cref{thm:upper} implies a dynamic algorithm for a special case of the dynamic covering SDP problem. 

We first define packing and covering semi-definite programs (SDPs).\footnote{Some papers \cite{jambulapati2021positive} flip our definition of packing and covering SDPs by considering their dual form.}

\begin{definition}[Packing/Covering SDP] Let $\CC, \AA_i$'s for $i = 1,2,...,m$ be $n\times n$ symmetric PSD matrices and $b_i$'s denote positive real numbers. The packing SDP problem asks to find
\[
\max_{\YY\succeq 0} \Tr[\CC\YY]  \quad \text{s.t.  } \Tr[\AA_i\YY]\leq b_i, \quad\forall i = 1,\cdots, m.
\]
The covering SDP problem asks to find
\[
\min_{\YY\succeq 0} \Tr[\CC\YY]  \quad \text{s.t.  } \Tr[\AA_i\YY]\geq b_i, \quad\forall i = 1,\cdots, m.
\]
\end{definition}
Note that when the matrices $\CC$ and $\AA_i$'s are all diagonal matrices, the packing and covering SDP problems above are precisely the well-studied packing and covering LP problems, respectively. Near-linear time algorithms for $(1+\eps)$-approximately solving packing and covering LPs are very well-studied in a long line of work, some of which are~\cite{allen2015nearly,allen2019nearly,wang2016unified,quanrud2020nearly}, and these problems have many applications, such as in graph embedding~\cite{plotkin1995fast}, approximation algorithms~\cite{luby1993parallel,trevisan1998parallel}, scheduling~\cite{plotkin1995fast}, to name a few.

\paragraph{Dynamic LPs.}
Near-optimal dynamic algorithms for packing and covering LPs were shown in~\cite{bhattacharya2023dynamic}. The paper studies two kinds of updates -- {\it restricting} and {\it relaxing} updates.
Restricting updates can only shrink the feasible region. In contrast, relaxing updates can only grow the feasible region.
In \cite{bhattacharya2023dynamic}, the authors gave a deterministic algorithm that can maintain a $(1+\epsilon)$-approximate solution to either packing and covering LPs that undergo only restricting updates or only relaxing updates in total time $\Otil(N/\epsilon^3 + t/\epsilon)$, where $N$ is the total number of nonzeros in the initial input and the updates, and $t$ is the number of updates. Hence, this is optimal up to logarithmic factors. 

A natural question is whether one can generalize the near-optimal dynamic LP algorithms with polylogarithmic overhead by \cite{bhattacharya2023dynamic} to work with SDPs since SDPs capture many further applications such as maximum cuts~\cite{iyengar2011approximating,klein1996efficient}, Sparse PCA~\cite{iyengar2011approximating}, sparsest cuts~\cite{iyengar2010feasible}, and  balanced separators~\cite{orecchia2012approximating}, among many others.

\paragraph{Static SDPs.}
Unfortunately, the algorithms for solving packing and covering SDPs are much more limited, even in the static setting. Near-linear time algorithms are known only for \emph{covering} SDPs when the cost matrix $\CC = \II$ is the identity \cite{peng2012faster,allen2016using}. 

The fundamental barrier to working with general psd matrix $\CC$ in covering SDPs is that it is as hard as approximating the minimum eigenvalue of $\CC$ (consider the program $\max_{\YY\succeq 0} \Tr[\CC\YY]$ such that $\Tr[\YY]\leq 1$). To the best of our knowledge, near-linear time algorithms for approximating the minimum eigenvalue assume a near-linear-time solver for $\CC$, i.e., to compute $\CC^{-1}\xx$ in the near-linear time given $\xx$. This can be done, for example, by applying the power method to $\CC^{-1}$. When $\CC^{-1}$ admits a fast solver, sometimes one can approximately solve a covering SDP fast, such as for spectral sparsification of graphs~\cite{lee2017sdp}, and the max-cut problem~\cite{arora2007combinatorial,trevisan2009max}.

For packing SDPs, there is simply no near-linear time algorithm known. An algorithm for approximately solving packing SDPs and even the generalization to mixed packing-covering SDPs was claimed in \cite{jambulapati2021positive}, but there is an issue in the convergence analysis even for pure packing SDPs. Fast algorithms for this problem, hence, remain open.

\paragraph{Dynamic SDPs.}
Since near-linear time static algorithms are prerequisites for dynamic algorithms with polylogarithmic overhead, we can only hope for a dynamic algorithm for covering SDPs when $\CC$ is an identity. 
Below, we will show that our algorithm for \Cref{thm:upper} implies a dynamic algorithm for maintaining the covering SDP solution when there is a single constraint and the updates are restricting. This follows because this problem is equivalence to \Cref{prob:dyn}.

We first define the dynamic covering problem with a single constraint under restricting updates.

\begin{problem}[Covering SDP with a Single Matrix Constraint under Restricting Updates]\label{prob:dynsdp} Given $\AA_0\succeq 0$, an accuracy parameter $\epsilon>0$, and an online sequence of vectors $\vv_1,\vv_2,\cdots,\vv_T$ that update $\AA_t \gets \AA_t -\vv_t\vv_t^{\top}$ such that $\AA_t\succeq 0$. The problem asks to explicitly maintain, for all $t$, an $(1+\eps)$-approximate optimal value $\nu_t$ of the SDP, i.e.,  $$\nu_t \leq (1+\epsilon)OPT_t \defeq {\min}_{\YY\succeq 0, \Tr[\AA_t\YY]\geq 1}  \Tr[\YY].$$
Furthermore, given a query request, return a matrix $\QQ^{(t)}$ where $\YY^{(t)} = \QQ^{(t)}{\QQ^{(t)}}^{\top} \in \mathbb{R}^{n\times n}$ is a $(1+\eps)$-approximate optimal solution, i.e., 
    $$\Tr[\AA_t\YY^{(t)}]\geq 1\text{ and }\Tr[\YY^{(t)}]\leq \left(1+\epsilon\right) OPT_t.$$
\end{problem}
\Cref{prob:dynsdp} is equivalent to \Cref{prob:dyn} in the following sense: given an algorithm for \Cref{prob:dyn}, we can obtain an algorithm for \Cref{prob:dynsdp} with the same total update time and optimal query time. Conversely, given an algorithm for \Cref{prob:dynsdp}, we can obtain an algorithm for \Cref{prob:dyn} in the \emph{eigenvalue-only} version with the same total update time.

\begin{proposition}\label{prop:equiv sdp}
The following holds,
\begin{enumerate}
    \item Given an algorithm for \Cref{prob:dyn} with update time $\mathcal{T}$, there is an algorithm for Problem~\ref{prob:dynsdp} with update time 
 $\mathcal{T}$ and query time $O(n)$, i.e., the time required to query $\QQ^{(t)}$ at any time $t$. 
    \item Given an algorithm for Problem~\ref{prob:dynsdp} with update time $\widetilde{\mathcal{T}}$, there is an algorithm for \Cref{prob:dyn} that only maintains the approximate eigenvalues with update time at most $\widetilde{\mathcal{T}}$.
\end{enumerate}
\end{proposition}
\begin{proof}
  Let us first characterize the solution to Problem~\ref{prob:dynsdp} at any $t$ for $\epsilon=0$. Since any feasible solution $\YY$ is PSD it can be characterized as, $\YY = \sum_{i=1}^n p_i \yy_i\yy_i^{\top}$ for some unit vectors $\yy_i$'s. Now $\Tr[\YY] = \sum_i p_i$ and
  \[
 1\leq  \Tr[\AA_t\YY] = \sum_{i=1}^n p_i \Tr[\AA_t \yy_i\yy_i^{\top}] = \sum_{i=1}^n p_i \yy_i^{\top}\AA_t\yy_i \leq \lambda_{\max}(\AA_t) \sum_{i=1}^n p_i.
  \]
  The above implies that for any feasible solution $\YY^{(t)}$, it must hold that $\sum_{i=1}^n p_i\geq 1/\lambda_{\max}(\AA_t)$, and equality holds iff $\yy_i$'s are the maximum eigenvector of $\AA_t$ for all $t$. Since the problem is asking to minimize $\sum_{i=1}^n p_i$, the solution must be $\YY^{(t)} = \frac{1}{\lambda_{\max}(\AA_t)}\uu\uu^{\top}$ where $\uu$ is the maximum eigenvector of $\AA_t$. Note that here $\QQ^{(t)} = \frac{1}{\sqrt{\lambda_{\max}(\AA_t)}}\uu$. 

  We now show the first part, by proving that at any $t$, given $\epsilon$, the solution to \Cref{prob:dyn} gives a solution to Problem~\ref{prob:dynsdp}. Let $\lambda_t$ and $\ww_t$ denote an $\epsilon/2$-approximate solution to \Cref{prob:dyn} for some $t$. Consider the solution $\QQ^{(t)} = \frac{1}{\sqrt{(1-\epsilon/2)\lambda_t}}\ww_t$ which gives $\YY^{(t)} = \frac{1}{(1-\epsilon/2)\lambda_t}\ww_t\ww_t^{\top}$. Then,
  \[\Tr[\AA_t\YY^{(t)}] = \frac{1}{(1-\epsilon/2)\lambda_t} \ww_t^{\top}\AA_t\ww_t \geq (1-\epsilon/2)\frac{\lambda_{\max}(\AA_t)}{(1-\epsilon/2)\lambda_{\max}(\AA)}\geq 1.\]
Therefore, $\YY^{(t)}$ satisfies the constraints of Problem~\ref{prob:dynsdp}. Next we look at the objective value, $\nu_t = \Tr[\YY^{(t)}] = \frac{1}{(1-\epsilon/2)\lambda_t} \leq \frac{1}{(1-\epsilon/2)^2\lambda_{\max}(\AA_t)} \leq (1+\epsilon) OPT_t.$ We can maintain $\QQ^{(t)}$ by just maintaining $\ww_t,\lambda_t$ which requires no extra time. The time required to obtain $\QQ^{(t)}$, which is the query time, from $\ww_t$ and $\lambda_t$ is at most $O(nnz(\ww_t))=O(n)$ and the value of $\nu_t$ can be obtained in $O(1)$ time. 

To see the other direction, consider the solution $\QQ^{(t)}$, $\nu_t$ of Problem~\ref{prob:dynsdp}. We can set $\lambda_t = \frac{1}{\nu_t}$. This implies that,
\[
\lambda_t = \frac{1}{\nu_t}\geq \frac{\lambda_{\max}(\AA_t)}{(1+\epsilon)} \geq (1-\epsilon)\lambda_{\max}(\AA_t),
\]
as required. In this case, we do not require any extra time.
\end{proof}

By plugging \Cref{thm:upper} into \Cref{prop:equiv sdp}, we conclude the following.
\begin{corollary} There is a randomized algorithm for \Cref{prob:dynsdp} under a sequence of $T$ restricting updates, that given $n, \AA_0$ and $\epsilon>1/n$ as input, with probability at least $1-1/n$ works against an oblivious adversary in total update time
\[
 O\left(\frac{\log^{3}n \log^6 \frac{n}{\epsilon}\log \frac{\lambda_{\max}(\AA_0)}{\lambda_{\max}(\AA_T)}}{\epsilon^{4}}\left(nnz(\AA_0) + \sum_{i=1}^T nnz(\vv_i)\right) \right),
 \]
and query time $O(n)$.
\end{corollary}

\section{Omitted Proofs}
\Decision*
\begin{proof}
We first estimate the maximum eigenvalue of $\AA_0$. This can be done via the standard power method, or running {\sc PowerMethod}$(\epsilon/4\log n, \AA_0)$ until Line~\ref{line:before case} and returning $\ww\in \WW$ such that $\ww= \arg\max_{\ww\in \WW}\ww^{\top}\AA_0\ww$. From Lemma~\ref{thm:StaticPowerMain}, with probability at least $1-1/n^{10}$, $\ww$ satisfies $\nu = \ww^{\top}\AA_0\ww\geq (1-\epsilon/\log n)\lambda_{\max}(\AA_0)$. Now, $\AA_0/\nu$ has maximum eigenvalue at most $1+\frac{\epsilon}{\log n}$. This procedure takes time at most $O\left(\frac{\log^2 n\log\frac{n}{\epsilon}}{\epsilon}\cdot nnz(\AA_0)\right)$.

   For the sequence of updates, let $t_1,t_2,\cdots,t_k$ denote the time steps where the maximum eigenvalue decreases by a factor of at least $1-\epsilon/\log n$. We begin with solving {\sc DecMaxEV}$(\epsilon,\frac{\AA_0}{\nu},\frac{\vv_1}{\sqrt{\nu}},\cdots, \frac{\vv_T}{\sqrt{\nu}})$ using $\mathcal{A}$.  The algorithm returns $\ww_1,\ww_2,\cdots,\ww_{t_1-1},${\sc False},$\cdots$. Now $\nu$ is an $\epsilon$-approximate maximum eigenvalue for $\AA_1,\cdots \AA_{t_1-1}$, and the vectors $\ww_1,\cdots,\ww_{t_1-1}$ are the required eigenvectors for $\AA_1,\cdots \AA_{t_1-1}$ respectively. We also know that, $\lambda_{\max}(\AA_{t_1})\leq \nu\left(1-\frac{\epsilon}{\log n}\right)$. The total time taken so far is, $O(\frac{\log^2 n\log\frac{n}{\epsilon}}{\epsilon}\cdot nnz(\AA_0) + \mathcal{T})$. 

   Since, $\lambda_{\max}\left(\frac{\AA_{t_1}}{\nu(1-\frac{\epsilon}{\log n})}\right)\leq 1$, again solve {\sc DecMaxEV}$(\epsilon,\frac{\AA_{t_1}}{\nu(1-\epsilon/\log n)},\frac{\vv_{t_1}}{\sqrt{\nu(1-\frac{\epsilon}{\log n})}},\cdots, \frac{\vv_T}{\sqrt{\nu(1-\frac{\epsilon}{\log n})}})$ using $\mathcal{A}$. This time, the algorithm returns $\ww_{t_1+1},\ww_{t_1+2},\cdots, \ww_{t_2-1},${\sc False},$\cdots$. We can now repeat this process starting at $t_2$ until $t_k$.

   The total number of calls to $\mathcal{A}$ can be bounded by the number of times the maximum eigenvalue decreases by a factor of $1-\frac{\epsilon}{\log n}$. This can happen at most $\frac{\log n}{\epsilon}\log \frac{\lambda_{\max}(\AA_0)}{\lambda_{\max}(\AA_T)}$ times. Therefore, the total time required is at most $O\left(\frac{\log^2 n\log\frac{n}{\epsilon}}{\epsilon}\cdot nnz(\AA_0) + \frac{\log n}{\epsilon}\log \frac{\lambda_{\max}(\AA_0)}{\lambda_{\max}(\AA_T)}\mathcal{T}\right)$.
\end{proof}

\subsection*{Power Method}

The following lemma proves that at least one initial random vector has a component along the eigenspace of the maximum eigenvalue.
\begin{lemma}\label{lem:PMHighComp}
Let $\uu_1,\cdots \uu_l$ denote the maximum eigenvectors of $\AA$ and let $\vv^{(0)} \in \mathbb{R}^n$ denote the vector with entries sampled indenpendently from $N(0,1)$, i.e., $\vv^{(0)} = \sum_{i=1}^n\alpha_i\uu_i$, where $\alpha_i \sim N(0,1)$, $\uu_i$'s are eigenvectors of $\AA$. Then, with probability at least $3/4$, $\sum_{i=1}^l\alpha_i^2 \geq \frac{1}{25}$.
\end{lemma}
\begin{proof}
Let $\lambda_1,\lambda_2,\cdots, \lambda_n$ denote the eigenvalues of $\AA$ in decreasing order corresponding to the eigenvectors $\uu_1,\cdots, \uu_n$.  Let $\uu$ be a vector in the span of $\uu_1,\cdots, \uu_l$. We can write, $\uu=\sum_{i=1}^l \beta_i\uu_i$ with $\|\beta\| = 1$. We note that the inner product $\langle \vv^{(0)},\uu\rangle = \sum_{i=1}^l \alpha_i \beta_i$ has mean $0$ and variance $\sum_i \Var(\alpha_i)\beta_i^2 = 1$ (since $\Var(\alpha_i) =1$ and $\sum_i \beta_i^2 = 1$), and is also a gaussian variable, and thus follows the distribution $N(0,1)$. Therefore from the standard gaussian tables, 
\[
\Pr\left[|\langle \vv^{(0)},\uu\rangle|\geq \frac{1}{5}\right] \geq \frac{3}{4}.
\]
Using Cauchy-Schwarz, $|\langle \vv^{(0)},\uu\rangle| \leq \left(\sum_{i=1}^l\alpha_i^2\right)^{1/2}$. As a result, we also have,
\[
\Pr\left[\left(\sum_{i=1}^l\alpha_i^2\right)^{1/2}\geq \frac{1}{5}\right] \geq \frac{3}{4}.
\] 
\end{proof}

\PowerMethod*
\begin{proof}

Algorithm~\ref{alg:PowerMethod} starts with a random vector $\vv^{(0,r)}$ (Line~\ref{algline:randomInit}). Let $l$ denote the dimension of the eigenspace corresponding to the maximum eigenvalue. We denote $\vv^{(0,r)} = \sum_i \alpha^{(r)}_i \uu_i$ and from Lemma~\ref{lem:PMHighComp}, we have that for all $r$, initial vectors,
\[
\Pr\left[\sum_{i=1}^l(\alpha^{(r)})_i^2 \geq \frac{1}{25}\right]\geq 3/4.
\]
We now analyze Algorithm~\ref{alg:PowerMethod}. After $K$ iterations, the algorithm computes,
\begin{align*}
\AA^K\vv^{(0,r)} & =\sum_{i =1 }^{n}\alpha^{(r)}_i \lambda_i^K \uu_i .
\end{align*}
  Note that, since $K = \frac{4\log \frac{n}{\epsilon}}{\epsilon}$, for every $\lambda_i \leq (1-\epsilon/4)\lambda_1,$ we have $\lambda_i^{2K} \leq \frac{\epsilon^{2}}{n^{2}}\lambda_1^{2K}$. Let $\tilde{l}$ denote the dimension of the space spanned by all eigenvectors of $\AA$ with corresponding eigenvalues at least $(1-\epsilon/4)\lambda_1$, we split the sum as,
\[
{\vv^{(0,r)}}^{\top}\AA^{2K+1}\vv^{(0,r)}  =\sum_{i =1 }^{\tilde{l}}(\alpha^{(r)})_i^2 \lambda_i^{2K+1} + \sum_{i >\tilde{l}}(\alpha^{(r)})_i^2 \lambda_i^{2K+1} .
\]
 Now,
 \begin{equation}\label{eq:pm1}
 (\ww^{(r)})^{\top}\AA\ww^{(r)} = \frac{(\AA^K\vv^{(0,r)})^{\top}\AA(\AA^K\vv^{(0,r)})}{\|\AA^K\vv^{(0,r)}\|^2}  = \frac{{\vv^{(0,r)}}^{\top}\AA^{2K+1}\vv^{(0,r)}}{\|\AA^K\vv^{(0,r)}\|^2}  \geq \frac{\sum_{i =1 }^{\tilde{l}}(\alpha^{(r)})^2_i \lambda_i^{2K+1} }{\|\AA^K\vv^{(0,r)}\|^2}.
 \end{equation}
We now want to bound $\|\AA^K\vv^{(0,r)}\|^2$. We again expand this quantity and use the bounds on $\lambda_i$'s to get, 
\[
\|\AA^K\vv^{(0,r)}\|^2 = \sum_{i =1 }^{n}(\alpha^{(r)})_i^2 \lambda_i^{2K} \leq \sum_{i =1 }^{\tilde{l}}(\alpha^{(r)})_i^2 \lambda_i^{2K} + \left(\frac{\epsilon}{n}\right)^{2}\lambda_1^{2K}\sum_{i>\tilde{l}}(\alpha^{(r)})_i^2 \leq \sum_{i =1 }^{\tilde{l}}(\alpha^{(r)})_i^2 \lambda_i^{2K}+\lambda_1^{2K}\left(\frac{\epsilon}{n}\right)^2\|\vv^{(0,r)}\|^2.
\]
Since $\vv^{(0,r)}$ has entries in $N(0,1)$, from \Cref{lem:NormG}, with probability at least $1-1/n^2$, $\|\vv^{(0,r)}\|^2 \leq 2n$ and as a result, with probability at least $1-1/n^2$, $\|\AA^K\vv^{(0,r)}\|^2 \leq \sum_{i =1 }^{\tilde{l}}(\alpha^{(r)})_i^2 \lambda_i^{2K}+\left(\frac{\epsilon}{n}\right)\lambda_1^{2K}$ when $n\geq 5$. Using this bound in \eqref{eq:pm1} and $\epsilon<1/20$,
\[
(\ww^{(r)})^{\top}\AA\ww^{(r)} \substack{(a)\\ \geq}  \frac{\lambda_{\tilde{l}}\sum_{i =1 }^{\tilde{l}}(\alpha^{(r)})^2_i \lambda_i^{2K}}{\sum_{i =1 }^{\tilde{l}}(\alpha^{(r)})^2_i \lambda_i^{2K} + \epsilon n^{-1}\lambda_1^{2K}} \substack{(b)\\ \geq} \frac{\lambda_{\tilde{l}}}{1+\frac{\epsilon n^{-1}}{\sum_{i=1}^l(\alpha^{(r)})_i^2} }\substack{(c)\\ \geq} \frac{(1-\epsilon/4)\lambda_1}{1+\epsilon/4}\substack{(d)\\\geq} (1-\epsilon/2)\lambda_1,
\]
with probability at least $3/5$. To see this, in $(a)$ we used that $\|\AA^K\vv^{(0,r)}\|^2 \leq \sum_{i =1 }^{\tilde{l}}(\alpha^{(r)})_i^2 \lambda_i^{2K}+\left(\frac{\epsilon}{n}\right)\lambda_1^{2K}$ with probability at least $1-1/n^2 $, in $(b)$ we used that $\sum_{i =1 }^{\tilde{l}}(\alpha^{(r)})^2_i \lambda_i^{2K} \geq \lambda_1^{2K}\sum_{i=1}^l (\alpha^{(r)})_i^2$, and in $(c)$ we used that $\sum_{i=1}^l(\alpha^{(r)})_i^2 \geq 1/25$ with probability at least $3/4$.

We next prove our concentration bound, i.e., for at least one $r \in \{1,\cdots, 10\log n\}$, $(\ww^{(r)})^{\top}\AA\ww^{(r)}\geq (1-\epsilon/2)\lambda_1$ with probability at least $1-1/n$. Observe that the probability that $(\ww^{(r)})^{\top}\AA\ww^{(r)}< (1-\epsilon/2)\lambda_1$ for one $r$ is at most $2/5$. Since we have $R = 10\log n$ values of $r$, the probability that all of them are such that $(\ww^{(r)})^{\top}\AA\ww^{(r)}< (1-\epsilon/2)\lambda_1$ is $\left(\frac{2}{5}\right)^{R} \leq \frac{1}{n^{10}}.$
Therefore, we must have $(\ww^{(r)})^{\top}\AA\ww^{(r)}\geq (1-\epsilon/2)\lambda_1$ for at least one $r$ with probability at least $1-1/n^{10}$.

It remains to prove the runtime. In every iteration, we are multiplying $\AA$ with a vector. This takes time $\nnz(\AA)$. The total number of iterations is $O(\log \frac{n}{\epsilon}/\epsilon)$ for every initial random vector $\vv^{(0,r)}$. We run the entire method $10\log n$ times, giving us a total runtime of $O(\nnz(\AA) \frac{\log^2 \frac{n}{\epsilon}}{\epsilon})$. In the end we compute $(\ww^{(r)})^{\top}\AA\ww^{(r)}$ for $r = 1,\cdots, 10\log n$. This takes an additional time of $O(\log n\cdot  nnz(\AA))$.

We now prove that for $i$ such that $\lambda_i \leq \lambda_1/2$, for some $r$, with probability at least $1-1/n^{10}$, ${\ww^{(r)}}^{\top}\uu_i \leq \frac{1}{n^8}\cdot \frac{\lambda_i}{\lambda_1}$. We use the same argument as above to get that for at least one $r$, $\sum_{i=1}^d(\alpha^{(r)})_i^2 \geq 1/25$ with probability at least $1-1/n^{10}$. For this $r$, and $\epsilon<1/20$,
\begin{align*}
(\ww^{(r)})^{\top}\uu_i & = \frac{\alpha_i\lambda_i^K}{\sqrt{\sum_{i=1}^n\alpha_i^2\lambda_i^{2K}}} \leq  \frac{\alpha_i\lambda_i^{K-1/2}\lambda_i^{1/2}}{\lambda_1^K\sqrt{\sum_{i=1}^d\alpha_i^2}}\leq 5\alpha_i \left(\frac{\lambda_i}{\lambda_1}\right)^{1/2}\left(\frac{1}{2}\right)^{K-1/2} \leq 5\alpha_i \left(\frac{\lambda_i}{\lambda_1}\right)^{1/2}(1-10\epsilon)^{K-1/2}.
\end{align*}
Now, $\alpha_i\leq 10\log n$ with probability at least $1-1/n^{10}$. Therefore, we have with probabiltiy at least $1-2/n^{10}$,
\[
\left[(\ww^{(r)})^{\top}\uu_i\right]^2 \leq 2500\frac{\lambda_i}{\lambda_1} \log^2 n (1-10\epsilon)^{2K-1} \leq \frac{2500 \epsilon^{10}\log n}{n^{10}}\frac{\lambda_i}{\lambda_1} \leq \frac{1}{n^8} \cdot \frac{\lambda_i}{\lambda_1}.
\]
\end{proof}

\end{document}